\documentclass[a4paper,10pt,twoside]{article}

\usepackage{amssymb}
\usepackage{amsthm}
\usepackage{amsmath}
\usepackage{enumerate}
\usepackage{hyperref}

\newtheorem{thm}{Theorem}
\newtheorem{lem}{Lemma}
\newtheorem{rem}{Remark}
\newtheorem{cor}{Corollary}

\newcommand*{\di}{\, \mathrm{d} }

\title{Pathwise approximation of Feynman path integrals using simple random walks}

\author{Tam\'as Szabados\footnote{Address:
Department of Mathematics, Budapest University of Technology and Economics, Muegyetem rkp. 3, H
ep, 5 em, Budapest, 1521, Hungary, e-mail: szabados@math.bme.hu, telephone: +36 1 463 1111/ext. 5905, fax: +36 1 463 1677} \\
Budapest University of Technology and Economics}

\begin{document}

\maketitle

\begin{abstract}
The aim of the presented research is to give a rigorous mathematical approach to Feynman path integrals based on strong (pathwise) approximations based on simple random walks.
\end{abstract}

%%%%%%%%%%%%%%%%%%%%%%%%%%%%%%%%%%%%%%%%%%%%%%%%%%%%%%%%%%%%%%%%%%%

\section{Introduction} \label{se:Intro}

The \emph{Schr\"odinger equation} describing the non-relativistic motion of a single particle with
mass $m=1$ (and $\hbar = 1$) is
\begin{equation}\label{eq:Sch}
\frac{1}{i} \frac{\partial \psi}{\partial t} = \frac12 \Delta \psi - V(x) \psi, \qquad \psi(0,x) =
g(x),
\end{equation}
where $x \in \mathbb{R}^d$, $t \ge 0$, and $\psi(t, \cdot) \in L^2(\mathbb{R}^d)$ is the \emph{complex probability amplitude} of the particle. The probability density of finding the particle at the point $x$ at time $t$ is the squared modulus of the complex amplitude, $|\psi(t,x)|^2$. The potential $V$ and the initial condition $g$ should fulfill suitable assumptions for the existence and uniqueness of a solution.

Based on physics intuitions, Richard Feynman \cite{FEY48} suggested that the solution of this
equation can be given in the form of a path integral :
\begin{equation}\label{eq:pathintc}
\psi(t,x)= \frac{1}{Z} \int_{\Omega^{x}[0,t]} \exp\left\{i \int_{0}^{t} \left( \frac12
\left(\frac{d\omega }{dt}\right)^2 - V(\omega (s)) \right) \di s  \right\} g(\omega (t)) \di \omega ,
\end{equation}
where $\Omega^{x}[0,t]$ is the set of all possible trajectories $\omega $ of the particle, starting
from $\omega (0)=x$, over the time interval $[0,t]$. The expression after the second integral is the Lagrangian: the kinetic energy minus the potential energy of the particle; its integral is the action integral, which should be extremal along the path of a particle in classical physics. The symbol $\di \omega $ is a  mathematically non-existing Lebesgue-type product measure $\prod_{0 \le s \le t} \di \omega (s)$ over the infinite dimensional vector space of trajectories. The normalizing constant $Z$ cannot have a well-defined finite value either. However, the starting point of Feynman's seminal paper \cite{FEY48} is a discrete time approximation that makes perfect sense. This paper tries to follow a similar approach.

Mark Kac \cite{KAC49,KAC51} realized that one can give Feynman's idea a rigorous mathematical meaning
for the Schr\"odinger-type real-valued differential equation
\begin{equation}\label{eq:Schr}
\frac{\partial \psi}{\partial t} = \frac12 \Delta \psi - V(x) \psi, \qquad \psi(0,x) = g(x),
\end{equation}
when the unknown function $\psi$ is real valued. From the physics point of view, here $\psi$ can be thought of as a function of an imaginary time $it$. In the path integral the exponential of the kinetic energy $\exp(-\frac12 (d\omega /dt)^2)$ can be moved into the measure that would become the Wiener measure $\mathbb{P}^x$ over the space of continuous trajectories $C[0,t]$ starting from the point $x$. This way one arrives at a rigorous path integral, the celebrated \emph{Feynman--Kac formula},
\begin{equation}\label{eq:pathintr}
\psi(t,x)= \int_{C[0,t]} \exp\left(- \int_{0}^{t} V(\omega (s)) \di s  \right) g(\omega (t)) \di
\mathbb{P}^x(\omega) ,
\end{equation}
which then really gives the unique solution of equation (\ref{eq:Schr}) when, for example, $V$ is bounded below, piecewise continuous, and $g$ is integrable.

There exists a huge literature that intend to give solid mathematical foundation to Feynman's original path integral (\ref{eq:pathintc}); here is a sample of some very significant ones:
\cite{GJ56,ITO60,CAM60,NEL64,CS83,ALB97,KOL05}. A good description of different rigorous approaches can be found in \cite{KOL05}.

The aim of the present paper is to give a rigorous approach to complex-valued path integrals based on a strong (pathwise) approximation by simple, symmetric random walks. In the real-valued case, weakly convergent (that is, convergent in distribution) approximations based on simple, symmetric random walks were given, for example, in \cite{KAC49,CSA93,NSZ16}. For simplicity, the paper is restricted to one spatial dimension, $d=1$, but could be extended to any finite dimension $d$.

The present work was helped by many numerical experiments and computations, using Wolfram Mathematica and Maple.

%%%%%%%%%%%%%%%%%%%%%%%%%%%%%%%%%%%%%%%%%%%%%%%%%%%%%%%%%%%%%%%%%%%

\section{Complex measure walk} \label{se:Complex}

Our first intention is to find a nearest neighbor, symmetric random walk approximation to the complex-valued case (\ref{eq:Sch}) and (\ref{eq:pathintc}). Fix a positive integer $n$ and take the measurable space $(\mathbb{R}^n, \mathcal{B}^n)$, where $\mathcal{B}^n$ denotes the Borel $\sigma$-field in $\mathbb{R}^n$. Take a sequence $(X_r)_{r=1}^n$, the steps of a symmetric nearest neighbor random walk; each $X_r$ has the set of possible values $\{-1, 0, 1\}$. Define partial sums by
\[
S^x_0 = x \in \mathbb{Z}, \qquad S^x_{k} = x+\sum_{r=1}^k X_r \qquad (1 \le k \le n) .
\]
When $x=0$, we simply write $S_n$.  Now the distribution of a step $X_r$ on $\{-1, 0, 1\}$ will be given by \emph{a complex measure} $\mu$ concentrated on $\{-1, 0, 1\}$:
\begin{equation}\label{eq:compl_ampl}
\mu(X_r = 1)= \mu(X_r = -1) = p \in \mathbb{C}, \qquad \mu(X_r = 0)= q \in \mathbb{C}.
\end{equation}
Presently, $p$ and $q$ are unknowns that we intend to determine below.

Since we want to have independent and identically distributed steps, the complex measure on $(\mathbb{R}^n, \mathcal{B}^n)$ corresponding to the sequence $(X_1, \dots, X_n)$ will be the $n$th power $\mu^n$. If it causes no ambiguity, this product measure will also be denoted by $\mu$.

Then $(S^x_k)_{k = 0}^{n}$ will be called a \emph{complex measure walk}. Notice that the existence of an infinite product measure is not claimed; consequently, our complex measure walks will have finite lengths. If $f : \mathbb{R}^n \to \mathbb{R}$ is an arbitrary Borel-measurable function, then the complex distribution of the random variable $Y = f(X_1, \dots , X_n)$ is determined by the standard rule
\[
\mu(Y \in A) := \mu^n\{(x_1, \dots , x_n) \in \mathbb{R}^n : f(x_1, \dots , x_n) \in A \} = \mu^n(f^{-1}(A)) .
\]
for any Borel set $A \in \mathcal{B}$. This rule determines the complex distribution of the random walk $S^x_k$ as well.

Denote expectations with respect to $\mu$ by $\mathbb{E}_{\mu}$. For example,
\begin{equation}\label{eq:compl_expect}
\mathbb{E}_{\mu} Y := \int_{\mathbb{R}^n} f(x_1, \dots, x_n) \di \mu ,
\end{equation}
if $Y = f(X_1, \dots, X_n)$ as above.

Now take two arbitrary $\mathbb{Z} \to \mathbb{R}$ functions $V$ and $g$. Then the complex difference equation corresponding to (\ref{eq:Sch}) is going to be
\begin{multline}\label{eq:differ_compl}
\frac{1}{i} \left(\psi(k+1,x) - \psi(k,x)\right)\\
 = \frac12 e^{-iV(x)} \left( \psi(k,x+1) - 2 \psi(k,x) + \psi(k,x-1)\right) \\
 + \frac{1}{i} \left(e^{-iV(x)}-1\right) \psi(k,x).
\end{multline}
The tentative solution of this difference equation is obtained by using the idea of Mark Kac to move the  exponential of the kinetic energy factor $\exp(i \frac12 (d\omega/dt)^2)$ (or its discrete time approximation) into the complex measure $\mu$:
\begin{equation}\label{eq:psik_compl}
 \psi(k,x) = \mathbb{E}_{\mu}\left\{\exp\left(- i \sum_{r=0}^{k-1} V(S^x_r)\right) g(S^x_k)\right\} \qquad (0 \le k \le n, x \in \mathbb{Z}).
\end{equation}

\begin{lem} \label{le:compl}
Under the previous conditions, (\ref{eq:psik_compl}) is a solution of the difference equation (\ref{eq:differ_compl}) if and only if $p=i/2$ and $q=1-i$.

Then (\ref{eq:psik_compl}) is the unique solution of the above difference equation under the initial condition
$\psi(0, x) = g(x)$  $(x \in \mathbb{Z})$ .
\end{lem}
\begin{proof}
Let us separate the first step of the complex measure walk:
\begin{multline}\label{eq:cdiscrecur}
\psi(k+1,x) = \mathbb{E}_{\mu}\left\{\exp\left(- i \sum_{r=0}^{k} V(S^x_r)\right) g(S^x_{k+1})\right\} \\
= e^{-iV(x)} \sum_{j\in\{-1,0,1\}} \mathbb{E}_{\mu}\left\{\exp\left(- i \sum_{r=0}^{k-1} V(S^{x+j}_r)\right) g(S^{x+j}_k) \right\} \mu(X_1=j) \\
= e^{-iV(x)} \left\{\psi(k, x-1) p + \psi(k, x) q + \psi(k, x+1) p \right\}.
\end{multline}

That is,
\begin{multline*}
\frac{1}{i} \left(\psi(k+1, x) - \psi(k,x) \right)
 = \frac{p}{i} e^{-iV(x)} \left\{\psi(k, x+1) -2 \psi(k,x) + \psi(k, x-1) \right\} \\
 + \frac{1}{i} \psi(k, x) \left((2p+q) e^{-iV(x)} - 1\right),
\end{multline*}
which is identical with (\ref{eq:differ_compl}) if and only if $p=i/2$ and $q=1-i$.

The uniqueness of this solution under the initial condition follows by induction over $k$ using the recursion (\ref{eq:cdiscrecur}) and starting with the initial condition.
\end{proof}

While this deduction of the complex distribution of the steps of the walk has been very simple, still its result has a huge impact on everything that follows afterwards. From now on, we always assume that $p=i/2$ and $q=1-i$. Observe the interesting fact that while the total variation of the complex measure $\mu$ for a single step is $1+\sqrt{2} > 1$ and so the total variation for the product measure with $n$ steps goes to $\infty$ as $n \to \infty$, still we have
\[
\mu(\mathbb{R}^n) = \mu(\{-1,0,1\}^n)=1 \qquad  \text{for any} \qquad n \ge 1 .
\]
The latter follows from the fact that $\mu(\mathbb{R}) = \mu(\{-1,0,1\}) = 1$ for a single step and $\mu$ was defined as its $n$th power for $n \ge 1$ steps.

Now let us determine the resulting complex law of the partial sums $(S_{\ell})_{\ell=0}^n$, $n \ge 1$, when the initial point is $S_0=0$. This deduction is based on the standard observation that a path from the origin to a point $j \in\mathbb{Z}$ in $\ell \ge 0$ steps is the result of a number (say, $r$) horizontal steps, and the difference of the remaining up and down steps must be $j$ ($|j| \le \ell$):
\begin{multline}\label{eq:phi_S_ell}
\mu(S_\ell=j) = \sum_{r=0}^{\ell-|j|} \binom{\ell}{r} \binom{\ell-r}{\frac{\ell-r+j}{2}} \left(\frac{i}{2}\right)^{\ell-r} (1-i)^r \\
= i^\ell 2^{-\ell} \ell! \sum_{r=0}^{\ell-|j|} \frac{(-1)^r 2^r (1+i)^r}{r! \, \frac{\ell-r-j}{2}! \, \frac{\ell-r+j}{2}!} .
\end{multline}
Here we used the convention that a term is $0$ whenever $\frac{\ell-r+j}{2}$ is not an integer. It follows from the above argument about the product measure that
\begin{equation}\label{eq:sumphi}
\sum_{j=-\ell}^{\ell} \mu(S_{\ell} = j) = 1 .
\end{equation}

Not surprisingly, formula (\ref{eq:phi_S_ell}) for $\mu(S_{\ell}=j)$ can be expressed in terms of (terminating) hypergeometric functions as well. Recall (see e.g. \cite{EMOT53}) that a hypergeometric function is defined by a series
\begin{equation}\label{eq:hypgeo}
{}_2 F_1(a,b;c;z) = \sum_{r=0}^{\infty} \frac{(a)_r (b)_r}{(c)_r} \frac{z^r}{r!},
\end{equation}
where we used the Pochhammer symbol $(x)_r:=x(x+1) \cdots (x+r-1)$ if $r > 0$; $(x)_r:=1$ if $r=0$.
\begin{lem} \label{le:hypgeo}
For $n \ge 0$  and $|m| \le n$ we have
\begin{equation*}
\mu(S_{2n}=2m)=  (-1)^n \binom{2n}{n+m} 2^{-2n} \; {}_2 F_1(-n-m,-n+m;\frac12;2i),
\end{equation*}
\begin{multline*}
\mu(S_{2n}=2m+1) \\
= (-1)^n (-2 -2i) (n-m) \binom{2n}{n+m} 2^{-2n} \; {}_2 F_1(-n-m,-n+m+1;\frac32;2i),
\end{multline*}
\begin{multline*}
\mu(S_{2n+1}=2m) \\
= i (-1)^n (-2 -2i) (n-m+1) \binom{2n+1}{n+m} 2^{-2n-1}  \; {}_2 F_1(-n-m,-n+m+1;\frac32;2i);
\end{multline*}
and for $-n-1 \le m \le n$,
\begin{multline*}
\mu(S_{2n+1}=2m+1) \\
= i (-1)^n \binom{2n+1}{n+m+1} 2^{-2n-1} \; {}_2 F_1(-n-m-1,-n+m;\frac12;2i).
\end{multline*}
\end{lem}
\begin{proof}
The proof is a simple algebraic computation with finite sums, so omitted.
\end{proof}

We have already seen that these complex measures are related to binomial probabilities. Next we want to find recursive formulas for $\mu(S_{\ell}=j)$. The method called Zeilberger's algorithm, see \cite[Chapter 6]{PWZ97}, can do exactly that for hypergeometric series with proper hypergeometric terms, like the one we have in formula (\ref{eq:phi_S_ell}). The second order linear recursions obtained with the Maple program EKHAD, included in the above book, are given next.

\begin{lem}\label{le:recurs}
For $n \ge 0$, with $|j| \le 2n$ fixed, we have the recursion
\begin{multline} \label{eq:even}
\mu(S_{2n+4}=j) \\
= (3+4i) \frac{(n+1)(n+2)(n+\frac12)(n+\frac32)(n+\frac74)}{(n-\frac{j}{2}+2)(n+\frac{j}{2}+2)(n-\frac{j}{2}+\frac32)
(n+\frac{j}{2}+\frac32)(n+\frac34)}
\mu(S_{2n}=j) \\
- (2+4i) \frac{(n+2)(n+\frac32)(n+\frac54)\left(n^2 + \frac{3}{20} j^2 + \frac52 n +\frac{51}{40} + i(\frac{4}{20} j^2 - \frac{1}{20})\right)} {(n-\frac{j}{2}+2)(n+\frac{j}{2}+2)(n-\frac{j}{2}+\frac32)(n+\frac{j}{2}+\frac32)(n+\frac34)} \\
\times \mu(S_{2n+2}=j).
\end{multline}
Similarly, for $n \ge 0$, with $|j| \le 2n+1$ fixed, we have the recursion
\begin{multline} \label{eq:odd}
\mu(S_{2n+5}=j) \\
= (3+4i) \frac{(n+1)(n+2)(n+\frac32)(n+\frac52)(n+\frac94)}{(n-\frac{j}{2}+2)(n+\frac{j}{2}+2)(n-\frac{j}{2}+\frac52)
(n+\frac{j}{2}+\frac52)(n+\frac54)} \\
\times \mu(S_{2n+1}=j) \\
- (2+4i) \frac{(n+2)(n+\frac52)(n+\frac74)\left(n^2 + \frac{3}{20} j^2 + \frac72 n +\frac{111}{40} + i(\frac{4}{20} j^2 - \frac{1}{20})\right)} {(n-\frac{j}{2}+2)(n+\frac{j}{2}+2)(n-\frac{j}{2}+\frac52)(n+\frac{j}{2}+\frac52)(n+\frac54)} \\
\times \mu(S_{2n+3}=j).
\end{multline}
\end{lem}

Since $\mu(S_k=j)=\mu(S_k=-j)$, it is enough to consider the cases when $j\ge 0$. For any $j\ge 0$ fixed, the above recursion formulas give the value of $\mu(S_k=j)$ for any $k$ if we start, depending on parities, with $\mu(S_j=j)$ and $\mu(S_{j+2}=j)$ or $\mu(S_{j+1}=j)$ and $\mu(S_{j+3}=j)$, respectively. In turn, the latter values can be determined by formula (\ref{eq:phi_S_ell}):
\begin{eqnarray}\label{eq:init_cond}
\mu(S_j=j) &=& \left(\frac{i}{2}\right)^j, \\
\mu(S_{j+1}=j) &=& (1-i) (j+1) \left(\frac{i}{2}\right)^j, \nonumber \\
\mu(S_{j+2}=j) &=& (j+2) \left(-\frac14 - (j+1)i\right) \left(\frac{i}{2}\right)^j , \nonumber \\
\mu(S_{j+3}=j) &=& (1-i) (j+2) (j+3) \left(-\frac14 -\frac{j+1}{3}i\right) \left(\frac{i}{2}\right)^j. \nonumber
\end{eqnarray}

Numerical simulations showed the important conjecture that for any $n\ge0$,
\begin{equation}\label{eq:coupl}
\frac{|\mu(S_{2n}=j)|^2}{\sum_{r=-2n}^{2n} |\mu(S_{2n}=r)|^2} \approx \binom{2n}{n+j} 2^{-2n} = \mathbb{P}(S_{2n}^*=2j) \qquad (|j| \le n),
\end{equation}
where $(S_{n}^*)_{n=0}^{\infty}$ is an ordinary simple, symmetric random walk. Interestingly, the approximation is good even for small values of $n$, see Table \ref{ta:bin_appr} below. Also, the left hand side of (\ref{eq:coupl}) is very close to 0 when $n < |j| \le 2n$. There is a similar fit between the normalized $|\mu(S_{2n+1}=j)|^2$ and $\mathbb{P}(S_{2n+1}^*=2j-1)$ ($-n \le j \le n+1$). Thus we define the non-negative quantities
\begin{equation}\label{eq:coupl_prob}
\mathbb{P}(S_{\ell}=j) := \frac{|\mu(S_{\ell}=j)|^2}{Z_{\ell}} \qquad (|j| \le \ell)
\end{equation}
as \emph{the coupled probabilities of the complex measure walk} $(S_{\ell})_{\ell \ge 0}$, where
\begin{equation}\label{eq:norm_fact}
Z_{\ell} := \sum_{j=-\ell}^{\ell} |\mu(S_{\ell}=j)|^2
\end{equation}
is the \emph{normalizing factor}. It is clear that (\ref{eq:coupl_prob}) defines a probability distribution for any $\ell \ge 0$.

Obviously, the magnitude of the normalizing factor $Z_{\ell}$ plays an important role in the following. To have an idea about it, let us consider the middle term $\mu(S_{2n}=0)$, or, rather, the approximate recursion for $a_n \approx \mu(S_{2n}=0)$ obtained from (\ref{eq:even}) when $j=0$ and $n$ is very large:
\begin{equation}\label{eq:midterm}
a_{n+2} = (3+4i) a_n - (2+4i) a_{n+1} .
\end{equation}
This is a second order linear difference equation with constant coefficients. By a standard method, one can search for the solution in the form $a_n = c q^n$; $c,q \in \mathbb{C}$. Then one gets the characteristic equation
$q^2 + (2+4i)q -(3+4i) = 0$ and the characteristic roots $q_1=1$ and $q_2=-3-4i$. Thus the general solution of (\ref{eq:midterm}) is
\[
a_n = c_1 + c_2 (-3-4i)^n \qquad (c_1, c_2 \in \mathbb{C}).
\]
The next section gives a more precise result.

We will need a recursion of the complex amplitudes $\mu(S_{\ell} = j)$ in the variable $j$ as well. For simplicity, we restrict ourselves to the case $\ell = 2n$ $(n \ge 1)$ fixed and $j = 2m$ ($m \ge 0$, integer). The other cases are similar. Define $a_{n,m} := \mu(S_{2n} = 2m)$.
\begin{lem}\label{le:recur_m}
If $n \ge 1$ fixed and $m$ goes from $n-1$ to $0$, one gets the reverse recursion
\begin{multline*}
a_{n,m} = \frac{\left(2 + \frac{1}{m + \frac12}\right) \left[\left(1 - \frac{m+\frac12}{n}\right) \left(1 + \frac{m+\frac32}{n}\right) - \frac{1}{2n}\right] + i8 \frac{m+1}{n} \frac{m+\frac32}{n}}{\left(1 + \frac{1}{m + \frac12}\right) \left(1 - \frac{m}{n}\right) \left(1 - \frac{m+\frac12}{n}\right)} \, a_{n,m+1} \\
- \frac{\left(1 + \frac{m+\frac32}{n}\right) \left(1 + \frac{m+2}{n}\right)}{\left(1 + \frac{1}{m + \frac12}\right) \left(1 - \frac{m}{n}\right) \left(1 - \frac{m+\frac12}{n}\right)} \, a_{n,m+2} ,
\end{multline*}
with the initial condition
\[
a_{n,n} = (-1)^n 2^{-2n}, \qquad a_{n,n+1} = 0 .
\]
\end{lem}
\begin{proof}
By Lemma \ref{le:hypgeo},
\[
a_{n,m} = (-1)^n \binom{2n}{n+m} 2^{-2n} \; {}_2 F_1(-n-m,-n+m;\frac12;2i) .
\]
Then, first,
\[
\frac{\binom{2n}{n+m+1}}{\binom{2n}{n+m}} = \frac{1 - \frac{m}{n}}{1 + \frac{m+1}{n}} .
\]
Next, a second order linear recursion can be obtained for ${}_2 F_1(-n-m,-n+m;\frac12;2i)$ in the variable $m$ by Gauss' contiguous relations. After simplification, we obtain the second order linear reverse recursion for $a_{n,m}$ in the variable $m \ge 0$, as stated in the lemma.
\end{proof}

The second order recursion of Lemma \ref{le:recur_m} can be reduced to a first order recursion by introducing the ratio $\rho_{n,m} := a_{n,m+1}/a_{n,m}$ when $m$ goes from $n-1$ to $0$:
\begin{multline}\label{eq:rho}
\rho_{n,m} = \frac{\left(1 + \frac{1}{m + \frac12}\right) \left(1 - \frac{m}{n}\right) \left(1 - \frac{m+\frac12}{n}\right)}{\mu_{n,m} - \left(1 + \frac{m+\frac32}{n}\right) \left(1 + \frac{m+2}{n}\right) \rho_{n,m+1}} , \\
\mu_{n,m} := \left(2 + \frac{1}{m + \frac12}\right) \left\{\left(1 - \frac{m+\frac12}{n}\right) \left(1 + \frac{m+\frac32}{n}\right) - \frac{1}{2n}\right\} \\
+ i8 \frac{m+1}{n} \frac{m+\frac32}{n} ,
\end{multline}
with the initial condition $\rho_{n,n} = 0$.

%%%%%%%%%%%%%%%%%%%%%%%%%%%%%%%%%%%%%%%%%%%%%%%%%%%%%%%%%%%%%%%%%%%

\section{An asymptotic solution of the recursion} \label{se:Asymptotic}

Now we would like to find de Moivre--Laplace-type approximations to the ``binomial formula'' (\ref{eq:phi_S_ell}) of $\mu(S_{\ell} = j)$ when $\ell \to \infty$. If we try to use the standard technique based on Stirling's formula, we run into serious difficulties. The reason is that here the complex terms are wildly oscillating as $\ell \to \infty$. Consequently, there are cancellations of bigger terms, and it is necessary to precisely handle essentially the whole range of terms $|j| \le \ell$ and not just the terms around the center, say $|j| \le \ell^{\frac23}$. That is why instead of the standard approach now we want to find an asymptotic solution of the recursions obtained in the previous section.

\begin{thm}\label{th:asympt_soln}
An asymptotic solution of the second order linear recursions (\ref{eq:even}), (\ref{eq:odd}) with initial conditions (\ref{eq:init_cond}) is
\begin{equation}\label{eq:approx_solution}
\mu(S_{\ell} = j) = \sqrt{\frac{2}{\ell}} \left(\widetilde{c}_{1,j} e^{i \frac{j^2}{2\ell}} + c_2 e^{i \pi j - (2+i) \frac{j^2}{2\ell}} (-3-4i)^{\frac{\ell}{2}} \right) + \widetilde{\epsilon}_{\ell,j} ,
\end{equation}
where $\widetilde{c}_{1,j}, c_2 \in \mathbb{C}$ are suitable constants $(\ell \ge 1, |j| \le \ell)$.
For the error term $\widetilde{\epsilon}_{\ell,j}$, we have
\begin{equation}\label{eq:error_term1}
\widetilde{\epsilon}_{\ell,j} = c_2 \sqrt{\frac{2}{\ell}} e^{i \pi j - (2+i) \frac{j^2}{2\ell}} (-3-4i)^{\frac{\ell}{2}} \left(z_0 \frac{j^4}{2 \ell^3} + z_1 \frac{2}{\ell} + \widetilde{h}_{\ell,j} \right),
\end{equation}
where $|\widetilde{c}_{1,j}| \le \widetilde{k}_1$, $|\widetilde{h}_{\ell,j}| \le \widetilde{k}_2 \ell^{-2/3}$ for any $|j| \le \frac13 \ell^{\frac23}$ and $\ell \ge 1$, with suitable constants $z_0, z_1 \in \mathbb{C}$ and $\widetilde{k}_1, \widetilde{k}_2 \in (0, \infty)$.
\end{thm}
\begin{proof}
The proof is long, but elementary.
We concentrate on the case of (\ref{eq:even}), where $\ell=2n$, even; but at some essential points we show to the differences that the case $\ell=2n+1$, (\ref{eq:odd}) causes.

Define $a_{n,m} := \mu(S_{2n} = 2m)$, $m:=\frac{j}{2} \in \frac12 \mathbb{Z}$, $\epsilon_{n,m} := \widetilde{\epsilon}_{2n, 2m}$, $c_{1,m} := \widetilde{c}_{1, 2m}$ and $h_{n,m} := \widetilde{h}_{2n, 2m}$. Since $a_{n,m}$ is an even function of $m$, it is enough to consider the case $0 \le m \le \frac13 n^{\frac23}$, that we assume from now on.

Dividing the numerators and denominators of (\ref{eq:even}) by $n^5$ when $n \ge 1$, we get
\begin{multline} \label{eq:exact_recursn}
a_{n+2,m} = (3+4i) \beta_{n,m} \, a_{n,m} - (2+4i) \gamma_{n,m} \, a_{n+1, m} \\
:= (3 + 4i) \frac{\left(1+\frac{1}{n}\right) \left(1+\frac{2}{n}\right) \left(1+\frac{1}{2n}\right) \left(1+\frac{3}{2n}\right) \left(1+\frac{7}{4n}\right)}{\left(1+\frac{2+m}{n}\right) \left(1+\frac{2-m}{n}\right) \left(1+\frac{\frac32+m}{n}\right) \left(1+\frac{\frac32-m}{n}\right) \left(1+\frac{3}{4n}\right)} a_{n,m} \\
- (2 + 4i) \frac{\left(1+\frac{2}{n}\right) \left(1+\frac{3}{2n}\right) \left(1+\frac{5}{4n}\right) \left(1 + \frac{3m^2}{5n^2} + \frac{5}{2n} + \frac{51}{40 n^2} + i\left(\frac{4m^2}{5n^2} - \frac{1}{20 n^2} \right) \right)}{\left(1+\frac{2+m}{n}\right) \left(1+\frac{2-m}{n}\right) \left(1+\frac{\frac32+m}{n}\right) \left(1+\frac{\frac32-m}{n}\right) \left(1+\frac{3}{4n}\right)} \\
\times a_{n+1,m} .
\end{multline}

A tentative asymptotic general solution of this second order linear recursion according to (\ref{eq:approx_solution}) is
\begin{equation}\label{eq:recurs_sol}
a'_{n,m} := \frac{1}{\sqrt{n}} \left(c_{1,m} e^{i \frac{m^2}{n}} + c_2 e^{i2 \pi m - (2+i) \frac{m^2}{n}} (-3 - 4i)^n \right) \qquad (c_{1,m}, c_2 \in \mathbb{C}) .
\end{equation}
By (\ref{eq:error_term1}) it is enough to show that
\begin{equation}\label{eq:eps_nm}
\epsilon_{n,m} =  a_{n,m} - a'_{n,m} = \frac{c_2}{\sqrt{n}} e^{i2 \pi m - (2+i) \frac{m^2}{n}} (-3-4i)^{n} \left(z_0 \frac{m^4}{n^3} + \frac{z_1}{n} + h_{n,m} \right) ,
\end{equation}
where $|c_{1,m}| \le k_1$ and $|h_{n,m}| \le k_2 n^{-\frac23}$ for any $0 \le m \le \frac13 n^{\frac23}$ and $n \ge 3$.

If we define $\bar{a}_{n,m} := \mu(S_{2n+1} = 2m)$, $m:=\frac{j}{2} \in \frac12 \mathbb{Z}$, then by (\ref{eq:odd}) we similarly obtain
\begin{multline} \label{eq:exact_recursn_odd}
\bar{a}_{n+2,m} = (3+4i) \bar{\beta}_{n,m} \, \bar{a}_{n,m} - (2+4i) \bar{\gamma}_{n,m} \, \bar{a}_{n+1, m} \\
:= (3 + 4i) \frac{\left(1+\frac{1}{n}\right) \left(1+\frac{2}{n}\right) \left(1+\frac{3}{2n}\right) \left(1+\frac{5}{2n}\right) \left(1+\frac{9}{4n}\right)}{\left(1+\frac{2+m}{n}\right) \left(1+\frac{2-m}{n}\right) \left(1+\frac{\frac52+m}{n}\right) \left(1+\frac{\frac52-m}{n}\right) \left(1+\frac{5}{4n}\right)} \bar{a}_{n,m} \\
- (2 + 4i) \frac{\left(1+\frac{2}{n}\right) \left(1+\frac{5}{2n}\right) \left(1+\frac{7}{4n}\right) \left(1 + \frac{3m^2}{5n^2} + \frac{7}{2n} + \frac{111}{40 n^2} + i\left(\frac{4m^2}{5n^2} - \frac{1}{20 n^2} \right) \right)}{\left(1+\frac{2+m}{n}\right) \left(1+\frac{2-m}{n}\right) \left(1+\frac{\frac52+m}{n}\right) \left(1+\frac{\frac52-m}{n}\right) \left(1+\frac{5}{4n}\right)} \\
\times \bar{a}_{n+1,m} .
\end{multline}

In this case too, a tentative asymptotic general solution of this second order linear recursion by (\ref{eq:approx_solution}) is
\begin{equation}\label{eq:recurs_sol_odd}
\bar{a}'_{n,m} := \frac{1}{\sqrt{n+\frac12}} \left(c_{1,m} e^{i \frac{m^2}{n+\frac12}} + c_2 e^{i2 \pi m - (2+i) \frac{m^2}{n+\frac12}} (-3 - 4i)^{n+\frac12} \right) .
\end{equation}

With $r=2$ and $\frac32$, use a Taylor expansion for the first four factors in the denominator of (\ref{eq:exact_recursn}) when $n \ge 3$:
\[
\left(1 + \frac{r \pm m}{n}\right)^{-1} = \sum_{k=0}^{\infty} (-1)^k \left(\frac{r \pm m}{n}\right)^k .
\]
Then after simplification, the product of the first four factors of the denominator is
\begin{multline*}
\left(\left(1 + \frac{2 + m}{n}\right) \left(1 + \frac{2 - m}{n}\right) \left(1 + \frac{\frac32 + m}{n}\right) \left(1 + \frac{\frac32 - m}{n}\right)\right)^{-1} \\
= 1 - \frac{7}{n} + 2 \frac{m^2}{n^2} + 3 \frac{m^4}{n^4} - 21 \frac{m^2}{n^3} + \frac{123}{4} \frac{1}{n^2} - \frac{217}{2} \frac{1}{n^3} + O\left(\frac{m^6}{n^6}\right) + O\left(\frac{1}{n^4}\right).
\end{multline*}

Expanding the product of the numerator of $\beta_{n,m}$ and $(1 + 3/(4n))^{-1}$ in (\ref{eq:exact_recursn}) and then simplifying, the result is
\[
1 + 6 n^{-1} + 13 n^{-2} + \frac{189}{16} n^{-3} + O(n^{-4}) .
\]
Consequently,
\begin{equation}\label{eq:beta_nm}
\beta_{n,m} = 1 - \frac{1}{n} + 2 \frac{m^2}{n^2} + 3 \frac{m^4}{n^4} - 9 \frac{m^2}{n^3} + \frac{7}{4 n^2} - \frac{51}{16 n^3} + O\left(\frac{m^6}{n^6}\right) + O\left(\frac{1}{n^4}\right).
\end{equation}

Similarly, the product of the numerator of $\gamma_{n,m}$ and $(1 + 3/(4n))^{-1}$ in (\ref{eq:exact_recursn}) is
\begin{multline*}
1 + \frac{13}{2n} + \left(\frac35 + i \frac45\right) \frac{m^2}{n^2} + \left(\frac{12}{5} + i \frac{16}{5}\right) \frac{m^2}{n^3} + \left(\frac{313}{20} - i \frac{1}{20}\right) \frac{1}{n^2} \\
+ \left(\frac{2641}{160} - i \frac{1}{5}\right) \frac{1}{n^3} + O\left(\frac{m^6}{n^6}\right) + O\left(\frac{1}{n^4}\right),
\end{multline*}
so we obtain
\begin{multline}\label{eq:gamma_nm}
\gamma_{n,m} = 1 - \frac{1}{2n} + \left(\frac{13}{5}+ i \frac45\right) \frac{m^2}{n^2} + \left(\frac{21}{5} + i \frac85\right) \frac{m^4}{n^4} - \left(\frac{49}{5} + i \frac{12}{5}\right) \frac{m^2}{n^3} \\
+ \left(\frac{9}{10} - i \frac{1}{20}\right) \frac{1}{n^2} + \left(-\frac{107}{160} + i \frac{23}{20}\right) \frac{1}{n^3} + O\left(\frac{m^6}{n^6}\right) + O\left(\frac{1}{n^4}\right).
\end{multline}
(Here the error terms $O(\cdot)$ are complex valued. Then the $O(\cdot)$ notation refers to both the real and imaginary parts.)

Combining (\ref{eq:exact_recursn}), (\ref{eq:beta_nm}) and (\ref{eq:gamma_nm}), the result is
\begin{equation}\label{eq:asympt_recurs}
a_{n+2,m} = (3+4i) \left(\beta'_{n,m} + \Delta \beta_{n,m}\right) a_{n,m} - (2+4i) \left(\gamma'_{n,m} + \Delta \gamma_{n,m}\right) a_{n+1, m} ,
\end{equation}
where
\begin{equation} \label{eq:beta_gamma_pr}
\beta'_{n,m}:= 1 - \frac{1}{n} + 2 \frac{m^2}{n^2} , \qquad
\gamma'_{n,m} := 1 - \frac{1}{2n} + \left(\frac{13}{5} + i \frac{4}{5} \right) \frac{m^2}{n^2} ,
\end{equation}
and for $0 \le m \le \frac13 n^{\frac23}$ and $n \ge 3$,
\begin{equation}\label{eq:Delta_beta}
\Delta \beta_{n,m} =  3 \frac{m^4}{n^4} - 9 \frac{m^2}{n^3} + \frac{7}{4 n^2} -\frac{51}{16 n^3} + O\left(\frac{m^6}{n^6}\right) + O\left(\frac{1}{n^4}\right)
\end{equation}
and
\begin{multline}\label{eq:Delta_gamma}
\Delta \gamma_{n,m} =  \left(\frac{21}{5} + i \frac85\right) \frac{m^4}{n^4} - \left(\frac{49}{5} + i \frac{12}{5}\right) \frac{m^2}{n^3} + \left(\frac{9}{10} - i \frac{1}{20}\right) \frac{1}{n^2} \\
+ \left(-\frac{107}{160} + i \frac{23}{20}\right) \frac{1}{n^3} + O\left(\frac{m^6}{n^6}\right) + O\left(\frac{1}{n^4}\right).
\end{multline}

When $\ell = 2n+1$, from (\ref{eq:recurs_sol_odd}), one obtains exactly the same $\bar{\beta}'_{n,m} = \beta'_{n,m}$ and $\bar{\gamma}'_{n,m} = \gamma'_{n,m}$ as above in the case of $\ell = 2n$. Then
\begin{equation}\label{eq:asympt_recurs_odd}
\bar{a}_{n+2,m} = (3+4i) \left(\bar{\beta}'_{n,m} + \Delta \bar{\beta}_{n,m}\right) \bar{a}_{n,m} - (2+4i) \left(\bar{\gamma}'_{n,m} + \Delta \bar{\gamma}_{n,m}\right) \bar{a}_{n+1, m} .
\end{equation}
On the other hand, slightly differently,
\begin{equation}\label{eq:Delta_beta_odd}
\Delta \bar{\beta}_{n,m} =  3 \frac{m^4}{n^4} - 11 \frac{m^2}{n^3} + \frac{9}{4 n^2} - \frac{83}{16 n^3} + O\left(\frac{m^6}{n^6}\right) + O\left(\frac{1}{n^4}\right)
\end{equation}
and
\begin{multline}\label{eq:Delta_gamma_odd}
\Delta \bar{\gamma}_{n,m} =  \left(\frac{21}{5} + i \frac85\right) \frac{m^4}{n^4} - \left(\frac{62}{5} + i \frac{16}{5}\right) \frac{m^2}{n^3} + \left(\frac{23}{20} - i \frac{1}{20}\right) \frac{1}{n^2} \\
+ \left(-\frac{431}{160} + i \frac{1}{5}\right) \frac{1}{n^3} + O\left(\frac{m^6}{n^6}\right) + O\left(\frac{1}{n^4}\right).
\end{multline}

Since these basic approximations are almost the same when $\ell = 2n$ or $\ell = 2n+1$, from now on, we deal exclusively with the even case $\ell = 2n$.

Using Taylor expansion again with $r=1$ and $2$, we have
\begin{multline}\label{eq:delta_nr}
\delta_{n,r} := \frac{1}{\sqrt{n+r}} - \frac{1}{\sqrt{n}} \left(1 - \frac{r}{2n}\right)
= \frac{1}{\sqrt{n}} \left\{\frac{3r^2}{8n^2} -\frac{5 r^3}{16 n^3} + \sum_{k=4}^{\infty} \binom{-\frac12}{k} \left(\frac{r}{n}\right)^k\right\} \\
= \frac{1}{\sqrt{n}} \left(\frac{3r^2}{8n^2} -\frac{5 r^3}{16 n^3} + O(n^{-4})\right).
\end{multline}
Also, with $\alpha = 2+i$ or $-i$,
\begin{equation}\label{eq:eta_nm}
e^{-\alpha \frac{m^2}{n+r}} = e^{-\alpha \frac{m^2}{n} \left(1 + \frac{r}{n} \right)^{-1}}
= e^{-\alpha \frac{m^2}{n}} \left(1 + \alpha r \frac{m^2}{n^2}\right) + \eta_{n,m,r}(\alpha) ,
\end{equation}
where
\begin{multline}\label{eq:eta_nmr}
\eta_{n,m,r}(\alpha) := e^{-\alpha \frac{m^2}{n}} \left\{e^{-\alpha m^2 \left(\frac{1}{n+r} - \frac{1}{n}\right)} - \left( 1 + \alpha r \frac{m^2}{n^2} \right) \right\}  \\
= e^{-\alpha \frac{m^2}{n}} \left\{ \alpha r \frac{m^2}{n^2} \left(\left(1 + \frac{r}{n}\right)^{-1} - 1 \right) + \frac{\alpha^2 r^2}{2} \frac{m^4}{n^4}\left(1 + \frac{r}{n}\right)^{-2} \right. \\
\left. + \sum_{k=3}^{\infty} \frac{1}{k!} \left(\alpha r \frac{m^2}{n^2}\right)^k \left(1 + \frac{r}{n}\right)^{-k} \right\} \\
= e^{-\alpha \frac{m^2}{n}} \left\{ -\alpha r^2 \frac{m^2}{n^3} + \frac{\alpha^2 r^2}{2} \frac{m^4}{n^4} + O\left(\frac{m^6}{n^6}\right)\right\} .
\end{multline}

Using these, for any value of $n$ and $m$, we further approximate $a'_{n+1,m}$ and $a'_{n+2,m}$ by
\begin{eqnarray}
a^{(1)}_{n,m} &:=& \frac{1}{\sqrt{n}} \left(1 - \frac{1}{2n}\right) \left\{c_{1,m} e^{i \frac{m^2}{n}}\left(1 - i \frac{m^2}{n^2}\right) \right. \nonumber \\
&& \left. + c_2 e^{i 2\pi m - (2+i) \frac{m^2}{n}} \left(1 + (2+i) \frac{m^2}{n^2}\right) (-3 - 4i)^{n+1} \right\}, \label{eq:a1nm} \\
a^{(2)}_{n,m} &:=& \frac{1}{\sqrt{n}} \left(1 - \frac{1}{n}\right) \left\{c_{1,m} e^{i \frac{m^2}{n}}\left(1 - 2i \frac{m^2}{n^2}\right) \right. \nonumber \\
&& \left. + c_2 e^{i 2 \pi m - (2+i) \frac{m^2}{n}} \left(1 + 2(2+i) \frac{m^2}{n^2}\right) (-3 - 4i)^{n+2} \right\} , \label{eq:a2nm}
\end{eqnarray}
respectively. Then it is easy to check that for any $n$, $m$, $c_{1,m}$, and $c_2$, we have the \emph{exact equality}
\begin{equation}\label{eq:aprpr}
a^{(2)}_{n,m} = (3+4i) \beta'_{n,m} \, a'_{n,m} - (2+4i) \gamma'_{n,m} \, a^{(1)}_{n, m}.
\end{equation}
This indicates that $a'_{n,m}$, defined by (\ref{eq:recurs_sol}), may really be an asymptotic solution of the recursion (\ref{eq:exact_recursn}).

We want to show (\ref{eq:eps_nm}), that is, we have to show that
\begin{equation}\label{eq:hnm}
|h_{n,m}| \le k_2 \, n^{-\frac23} \qquad (n \ge 3, |m| \le \frac13 n^{\frac23})
\end{equation}
holds with suitable constants $z_0, z_1 \in \mathbb{C}$ and $k_2 \in (0, \infty)$.

Consequently, by (\ref{eq:delta_nr})--(\ref{eq:eta_nmr}) and a Taylor expansion of $(1+1/n)^k$ $(k=-3, -1)$ we should have
\begin{multline}\label{eq:eps_nm1}
\epsilon_{n+1,m} = \frac{c_2}{\sqrt{n}} e^{i 2 \pi m - (2+i) \frac{m^2}{n}} (-3-4i)^{n+1} \left(1 - \frac{1}{2n} + \frac{3}{8n^2} -\frac{5}{16 n^3} + O\left(\frac{1}{n^{4}}\right)\right) \\
\times \left(1 + (2+i) \frac{m^2}{n^2} + \frac{3+4i}{2} \frac{m^4}{n^4} - (2+i) \frac{m^2}{n^3} + O\left(\frac{m^6}{n^6}\right)\right) \\
\times \left\{z_0 \frac{m^4}{n^3} \left(1 - \frac{3}{n} + O\left(\frac{1}{n^2}\right) \right)+ \frac{z_1}{n} \left(1 - \frac{1}{n} + \frac{1}{n^2} + O\left(\frac{1}{n^3}\right) \right)+ h_{n+1,m} \right\} ,
\end{multline}
where
\begin{equation}\label{eq:hn1m}
|h_{n+1,m}| \le k_2 \, (n+1)^{-\frac23},
\end{equation}
with the above constants $z_0, z_1 \in \mathbb{C}$ and $k_2 > 0$.

Similarly, we should have
\begin{multline}\label{eq:eps_nm2}
\epsilon_{n+2,m} = \frac{c_2}{\sqrt{n}} e^{i 2 \pi m - (2+i) \frac{m^2}{n}} (-3-4i)^{n+2} \left(1 - \frac{1}{n} + \frac{3}{2n^2} -\frac{5}{2 n^3} + O\left(\frac{1}{n^{4}}\right)\right) \\
\times \left(1 + 2(2+i) \frac{m^2}{n^2} + 2(3+4i) \frac{m^4}{n^4} - 4(2+i) \frac{m^2}{n^3} + O\left(\frac{m^6}{n^6}\right)\right) \\
\times \left\{z_0 \frac{m^4}{n^3} \left(1 - \frac{6}{n} + O\left(\frac{1}{n^2}\right) \right) + \frac{z_1}{n} \left(1 - \frac{2}{n} + \frac{4}{n^2} + O\left(\frac{1}{n^3}\right) \right) + h_{n+2,m} \right\}  ,
\end{multline}
where
\begin{equation}\label{eq:hn2m}
|h_{n+2,m}| \le k_2 \, (n+2)^{-\frac23} ,
\end{equation}
with the above constants $z_0, z_1 \in \mathbb{C}$ and $k_2 > 0$.

On the other hand, by definition,
\begin{equation}\label{eq:eps_n2}
\epsilon_{n+2,m} :=  a_{n+2,m} - a'_{n+2,m} = (a_{n+2,m} - a^{(2)}_{n,m}) - (a'_{n+2,m} - a^{(2)}_{n,m}) .
\end{equation}
Let us compute the terms here. First, by (\ref{eq:recurs_sol})--(\ref{eq:eta_nmr}) and (\ref{eq:a2nm}),
\begin{multline}\label{eq:anmpr_an2}
a'_{n+2,m} - a^{(2)}_{n,m} = \frac{c_2}{\sqrt{n}} e^{i 2\pi m - (2+i) \frac{m^2}{n}} (-3-4i)^{n+2} \\
\times \left\{(6+8i) \frac{m^4}{n^4} - (8+4i) \frac{m^2}{n^3} + \frac{3}{2n^2} - \frac{5}{2 n^3}
+ O\left(\frac{m^6}{n^6}\right) + O\left(\frac{1}{n^4}\right)\right\} .
\end{multline}
Similarly,
\begin{multline}\label{eq:anmpr_an1}
a'_{n+1,m} - a^{(1)}_{n,m} = \frac{c_2}{\sqrt{n}} e^{i 2 \pi m - (2+i) \frac{m^2}{n}} (-3-4i)^{n+1} \\
\times \left\{\left(\frac32+2i\right) \frac{m^4}{n^4} - (2+i) \frac{m^2}{n^3} + \frac{3}{8n^2} - \frac{5}{16 n^3} + O\left(\frac{m^6}{n^6}\right) + O\left(\frac{1}{n^3}\right)\right\} .
\end{multline}
Next, by (\ref{eq:asympt_recurs}) and (\ref{eq:aprpr}),
\begin{multline}\label{eq:an2_a2n}
a_{n+2,m} - a^{(2)}_{n,m} \\
= (3+4i)(\beta_{n,m} a_{n,m} - \beta'_{n,m} a'_{n,m}) - (2+4i) (\gamma_{n,m} a_{n+1,m} - \gamma'_{n,m} a^{(1)}_{n,m}) \\
= (3+4i) \left(\Delta \beta_{n,m} \, a'_{n,m} + \beta_{n,m} \epsilon_{n,m}\right) \\
- (2+4i) \left\{\Delta \gamma_{n,m} \, a^{(1)}_{n,m} + \gamma_{n,m} \left(\epsilon_{n+1,m} + (a'_{n+1,m} - a^{(1)}_{n,m})\right)\right\}
\end{multline}
Thus, by (\ref{eq:asympt_recurs})--(\ref{eq:Delta_gamma}), (\ref{eq:eps_nm}), (\ref{eq:eps_nm1}) and (\ref{eq:eps_n2})--(\ref{eq:an2_a2n}), after simplification we obtain that
\begin{multline}\label{eq:eps_n2m}
\epsilon_{n+2,m} = \frac{c_2}{\sqrt{n}} e^{i 2\pi m - (2+i) \frac{m^2}{n}} (-3 - 4i)^{n} \\
\times \left\{z_0 \left((-7+24i)\frac{m^4}{n^3} + (76-82i)\frac{m^6}{n^5} +(37-84i)\frac{m^4}{n^4} - (248-136i)\frac{m^8}{n^7} \right) \right. \\
\left. + z_1 \left((-7+24i)\frac{1}{n} + (17-44i)\frac{1}{n^2} + (-29+78i)\frac{1}{n^3} + (-76+82i)\frac{m^2}{n^3} \right) \right. \\
\left. + h_{n,m} \left((3+4i) - (3+4i)\frac{1}{n} + \left(\frac{21}{4}+7i\right)\frac{1}{n^2} - \left(\frac{153}{16}+\frac{51}{4}i\right)\frac{1}{n^3} \right. \right. \\
\left. \left. + (6+8i)\frac{m^2}{n^2} + (9+12i) \frac{m^4}{n^4} - (27+36i)\frac{m^2}{n^3}\right)\right. \\
\left. + h_{n+1,m} \left((-10+20i) + (10-20i)\frac{1}{n} - \left(\frac{57}{4}-31i\right)\frac{1}{n^2} - \left(\frac{117}{16}+\frac{353}{8}i\right)\frac{1}{n^3} \right. \right. \\
\left. \left. - (82-74i)\frac{m^2}{n^2} - (257-124i)\frac{m^4}{n^4} + (247-254i)\frac{m^2}{n^3}\right)\right. \\
\left. + (4-3i)\frac{1}{n^2} - \left(\frac{275}{8}-\frac{25}{8}i\right)\frac{1}{n^3} + (114+2i)\frac{m^4}{n^4} +(7-74i)\frac{m^2}{n^3}  \right. \\
\left. + O\left(\frac{m^6}{n^6}\right) + O\left(\frac{1}{n^4}\right) \right\} .
\end{multline}

Compare the $m^4/n^4$ and the $1/n^2$ terms, respectively, in (\ref{eq:eps_nm2}) and (\ref{eq:eps_n2m}). We have equality if and only if
\begin{equation}\label{eq:z0_z1}
z_0 = \frac{1+8i}{6}, \qquad z_1 = \frac{1+i}{8} .
\end{equation}
Then express $h_{n+2,m}$ by comparing (\ref{eq:eps_nm2}) and (\ref{eq:eps_n2m}) again:
\begin{multline*}
h_{n+2,m} \\
= \frac{1}{3+4i} \, \frac{1 - \frac{1}{n} + \frac{7}{4 n^2} + 2 \frac{m^2}{n^2} - \frac{51}{16 n^3} + 3 \frac{m^4}{n^4} - 9 \frac{m^2}{n^3}}{1 - \frac{1}{n} + \frac{3}{2 n^2} + (4+2i)\frac{m^2}{n^2} - \frac{5}{2 n^3} + (6+8i)\frac{m^4}{n^4} - (12+6i)\frac{m^2}{n^3}} h_{n,m} \\
+ \frac{2+4i}{3+4i} \, \frac{1 - \frac{1}{n} + \frac{61-2i}{40 n^2} + \frac{23+9i}{5}\frac{m^2}{n^2} - \frac{259-188i}{160 n^3} + \frac{101+78i}{10}\frac{m^4}{n^4} - \frac{151+48i}{10}\frac{m^2}{n^3}}{1 - \frac{1}{n} + \frac{3}{2 n^2} + (4+2i)\frac{m^2}{n^2} - \frac{5}{2 n^3} + (6+8i)\frac{m^4}{n^4} - (12+6i)\frac{m^2}{n^3}} h_{n+1,m} \\
+ \frac{\frac{-19+317i}{400 n^3} - \frac{73-14i}{25}\frac{m^2}{n^3} + \frac{1+8i}{3}\frac{m^8}{n^7} + O\left(\frac{m^6}{n^6}\right) + O\left(\frac{1}{n^4}\right)}{1 - \frac{1}{n} + \frac{3}{2 n^2} + (4+2i)\frac{m^2}{n^2} - \frac{5}{2 n^3} + (6+8i)\frac{m^4}{n^4} - (12+6i)\frac{m^2}{n^3}} .
\end{multline*}
Simplifying this, we obtain that
\begin{multline}\label{eq:hn2mrec}
h_{n+2,m} = \frac{1}{3+4i} \,\left(1 + \frac{1}{4 n^2} - (2+2i)\frac{m^2}{n^2} + O\left(\frac{m^4}{n^4}\right) + O\left(\frac{1}{n^3}\right)\right) h_{n,m} \\
+ \frac{2+4i}{3+4i} \,\left(1 + \frac{1-2i}{40 n^2} - \frac{3-i}{5}\frac{m^2}{n^2} + O\left(\frac{m^4}{n^4}\right) + O\left(\frac{1}{n^3}\right)\right) h_{n+1,m} \\
+ \frac{-19+317i}{400 n^3} - \frac{73-14i}{25}\frac{m^2}{n^3} + \frac{1+8i}{3}\frac{m^8}{n^7} + O\left(\frac{m^6}{n^6}\right) + O\left(\frac{1}{n^4}\right) .
\end{multline}

Ignoring lower order error terms, now we want to find an asymptotic solution of the inhomogeneous linear difference equation
\begin{equation}\label{eq:hnmrec1}
h'_{n+2,m} = \mu_{n,m} h'_{n,m} + \nu_{n,m} h'_{n+1,m} + g_{n,m} ,
\end{equation}
where
\begin{eqnarray}\label{eq:hn2mcoef}
\mu_{n,m} &:=& \frac{1}{3+4i} \,\left(1 + \frac{1}{4 n^2} - (2+2i)\frac{m^2}{n^2}\right), \nonumber \\
\nu_{n,m} &:=& \frac{2+4i}{3+4i} \,\left(1 + \frac{1-2i}{40 n^2} - \frac{3-i}{5}\frac{m^2}{n^2}\right), \nonumber \\
g_{n,m} &:=& \frac{-19+317i}{400 n^3} - \frac{73-14i}{25}\frac{m^2}{n^3} + \frac{1+8i}{3}\frac{m^8}{n^7} .
\end{eqnarray}
Similarly as in (\ref{eq:recurs_sol}), we can conclude that an asymptotic general solution (ignoring lower order error terms again) of the corresponding homogeneous linear difference equation is
\begin{multline*}
h^{\text{hom}}_{n,m} := d_{1,m} u_{1,m}(n) + d_{2,m} u_{2,m}(n) , \\
u_{1,m}(n) = e^{-\frac{1}{n}(\alpha_1 m^2 + \beta_1)} , \qquad  u_{2,m}(n) = e^{-\frac{1}{n}(\alpha_2 m^2 + \beta_2)} (-3-4i)^{-n},
\end{multline*}
where $d_{1,m}, d_{2,m} \in \mathbb{C}$. Comparing with (\ref{eq:hn2mcoef}), it follows that
\begin{equation}\label{eq:hnmhom}
h^{\text{hom}}_{n,m} = d_{1,m} e^{\frac{-1+i}{16 n}} + d_{2,m} e^{(2+2i)\frac{m^2}{n} - \frac{3+i}{16 n}} (-3-4i)^{-n} .
\end{equation}

Then we are going to find an (asymptotic) particular solution of the inhomogeneous difference equation (\ref{eq:hnmrec1}). By variation of parameters, see e.g. \cite[Section 3.2]{KEPE91}, we seek a solution as
\[
h^{\text{inh}}_{n,m} = f_{1,m}(n) u_{1,m}(n) + f_{2,m} u_{2,m}(n) ,
\]
where $f_{j,m}(n)$ $(j=1,2)$ are unknown functions. Then one obtains a linear system of equations for the differences $\Delta f_{j,m}(n) := f_{j,m}(n+1) - f_{j,m}(n)$:
\begin{eqnarray*}
\Delta f_{1,m}(n) u_{1,m}(n+1) + \Delta f_{2,m}(n) u_{2,m}(n+1) &=& 0 \\
\Delta f_{1,m}(n) u_{1,m}(n+2) + \Delta f_{2,m}(n) u_{2,m}(n+2) &=& g_{n,m} .
\end{eqnarray*}
Thus a particular solution is
\begin{multline}\label{eq:hnminh}
h^{\text{inh}}_{n,m} = \sum_{k=1}^n g_{k,m} \frac{e^{\frac{1-i}{16} \left(\frac{1}{k+2} - \frac{1}{n}\right)}}{1 + (3+4i)^{-1} e^{-\frac{(2+2i)\left(m^2 - \frac{1}{16}\right)}{(k+1)(k+2)}}} \\
+ \sum_{k=1}^n g_{k,m} \frac{e^{\left(-(2+2i)m^2 + \frac{3+i}{16}\right) \left(\frac{1}{k+2} - \frac{1}{n}\right)} (-3-4i)^{k-n+2}}{1 + (3+4i) e^{\frac{(2+2i)\left(m^2 - \frac{1}{16}\right)}{(k+1)(k+2)}}} .
\end{multline}

Lemma \ref{le:hinf} below shows that $h^{\text{inh}}_{n,m}$ uniformly converges to a finite complex limit for any $|m| \le \frac13 n^{\frac23}$ as $n \to \infty$:
\begin{equation}\label{eq:hinfty}
h^{\text{inh}}_{\infty,m} := \lim_{n \to \infty} h^{\text{inh}}_{n,m} = \sum_{k=1}^{\infty} g_{k,m} \frac{e^{\frac{1-i}{16(k+2)}}}{1 + (3+4i)^{-1} e^{-\frac{(2+2i)\left(m^2 - \frac{1}{16}\right)}{(k+1)(k+2)}}} .
\end{equation}
For example, $h^{\text{inh}}_{\infty,0} \approx -0.156498 + 0.849953 i$.

Thus an asymptotic solution of the difference equation (\ref{eq:hnmrec1}) is
\begin{equation}\label{eq:hdpnm}
h''_{n,m} := h^{\text{hom}}_{n,m} + h^{\text{inh}}_{n,m} ,
\end{equation}
as defined by (\ref{eq:hnmhom}) and (\ref{eq:hnminh}). Let us replace $h_{n,m}$ by $h''_{n,m}$ in (\ref{eq:eps_nm}), since this way we can explicitly compute approximate values of coefficients.

For $1 \le m \le \frac13 n^{\frac23}$, we choose the coefficient $c_{1,m}$ so that $a'_{m,m} = a_{m,m} = (-1)^m 2^{-2m}$. (So then the error term be zero.) It implies that
\[
c_{1,m} = (-1)^m 2^{-2m} \sqrt{m} e^{-im} - c_2 e^{-(2+2i)m} (-3-4i)^m .
\]
When $m = 0$, we have to modify this, since $a'_{0,0}$ is not defined. Then we demand that $a'_{1,0} = a_{1,0} = -\frac12 - 2i$. It gives
\[
c_{1,0} = -\frac12 - 2i + (3 + 4i) c_2.
\]
From these it follows that there is a positive constant $k_1$ such that $|c_{1,m}| \le k_1$ for any $0 \le m \le \frac13 n^{\frac23}$, as was claimed.

To find the coefficients $d_{1,m}$ and $d_{2,m}$ in (\ref{eq:hnmhom}) we demand two boundary conditions. First, we want that the (approximate) error term
\begin{multline}\label{eq:eps_nm_appr}
\epsilon''_{n,m} :=  \frac{c_2}{\sqrt{n}} e^{- (2+i) \frac{m^2}{n}} (-3-4i)^{n} \\
\times \left(z_0 \frac{m^4}{n^3} + \frac{z_1}{n} + d_{1,m} e^{\frac{-1+i}{16 n}} + d_{2,m} e^{(2+2i)\frac{m^2}{n} - \frac{3+i}{16 n}} (-3-4i)^{-n} + h^{\text{inh}}_{n,m} \right)
\end{multline}
be $0$ when $n=m$ $(m \ge 1)$ and $n=1$ $(m=0)$. Second, we demand that the factor in parentheses above tend to $0$ as $n \to \infty$. From the second condition it follows that
\begin{equation}\label{eq:d1m}
d_{1,m} = -h^{\text{inh}}_{\infty,m} \qquad \left(0 \le m \le \frac13 n^{\frac23}\right).
\end{equation}

From the first condition it follows that
\begin{equation}\label{eq:d2m}
d_{2,m} = -\left(z_0 m + \frac{z_1}{m} + d_{1,m} e^{\frac{-1+i}{16 m}} + h^{\text{inh}}_{m,m}\right) (-3-4i)^m e^{-(2+2i)m + \frac{3+i}{16 m}}
\end{equation}
$(1 \le m \le \frac13 n^{\frac23})$ and
\begin{equation}\label{eq:d20}
d_{2,0} = \left(z_1 + d_{1,0} e^{\frac{-1+i}{16}} + h^{\text{inh}}_{1,0}\right) (3+4i) e^{\frac{3+i}{16}} \qquad (m = 0) .
\end{equation}

(\ref{eq:d1m}) and Lemma \ref{le:hinf} below show that
\[
|d_{1,m}| \le \kappa_7 (1 + m^8)  \qquad (m \in \mathbb{Z}).
\]
Hence (\ref{eq:d2m}), (\ref{eq:d20}) and Lemma \ref{le:hinf} imply that $|d_{2,m}| \to 0$ as $m \to \infty$. Consequently, there exists a positive constant $k_4$ such that $|d_{2,m}| \le k_4$ for any $m \in \mathbb{Z}$.

Since the error term $h_{n,m}$ is asymptotically equal to $h''_{n,m}$, by these and Lemma \ref{le:hinf} we conclude that
\begin{multline*}
|h_{n,m}| \sim |h''_{n,m}| = \left| -h^{\text{inh}}_{\infty,m} e^{\frac{1-i}{16 n}} + d_{2,m} e^{(2+2i)\frac{m^2}{n} - \frac{3+i}{16 n}} (-3-4i)^{-n} + h^{\text{inh}}_{n,m} \right| \\
\le \kappa_8 n^{-\frac23} + k_4 \, e^{\frac29 n^{1/3}} 5^{-n} \le k_2 \, n^{-\frac23}
\end{multline*}
for any $|m| \le \frac13 n^{\frac23}$, with some positive constant $k_2$. This completes the proof of the theorem.

\end{proof}

\begin{lem}\label{le:hinf}
If $h^{\text{inh}}_{n,m}$ is defined by (\ref{eq:hnminh}) and $h^{\text{inh}}_{\infty,m}$ is defined by (\ref{eq:hinfty}), then
\[
|h^{\text{inh}}_{\infty,m}| < \kappa_7 (1 + m^8) \quad  (m \in \mathbb{Z}),
\]
\[
|h^{\text{inh}}_{\infty,m} - h^{\text{inh}}_{n,m}| \le \kappa_8 \, n^{-\frac23}  \quad (|m| \le \frac13 n^{\frac23}, n \ge 1) ,
\]
which  converges to $0$, uniformly in $m$, as $n \to \infty$. Finally,
\[
|h^{\text{inh}}_{\infty,m} - h^{\text{inh}}_{m,m}| \le \kappa_9 (1 + m^2) \quad  (m \in \mathbb{Z}) .
\]
Here and in the proof, each $\kappa_r$ is a finite, positive constant.
\end{lem}
\begin{proof}
Let us begin with the second sum in (\ref{eq:hnminh}). We are going to show that the limit of the second sum is $0$ when $|m| \le \frac13 n^{\frac23}$:
\begin{multline*}
\sum_{k=1}^n \left|g_{k,m} \frac{e^{\left(-(2+2i)m^2 + \frac{3+i}{16}\right) \left(\frac{1}{k+2} - \frac{1}{n}\right)} (-3-4i)^{k-n+2}}{1 + (3+4i) e^{\frac{(2+2i)\left(m^2 - \frac{1}{16}\right)}{(k+1)(k+2)}}} \right| \\
\le \kappa_1 \sum_{k=1}^n
\left(\frac{1+m^2}{k^3} + \frac{m^8}{k^7}\right) e^{\left(-2m^2 + \frac{3}{16}\right)\frac{n-2-k}{n(k+2)}} (-3-4i)^{k-n} \\
\le \kappa_1 e^{\frac{1}{16} + n^{-\frac23}} \sum_{k=1}^n
\left(\frac{1+m^2}{k^3} + \frac{m^8}{k^7}\right) (-3-4i)^{k-n} \\
\le \kappa_2 (-3-4i)^{-\frac{n}{2}} \sum_{k=1}^{\lfloor\frac{n}{2}\rfloor}
\left(\frac{1+m^2}{k^3} + \frac{m^8}{k^7}\right)  + \kappa_2 \sum^n_{k=\lfloor\frac{n}{2}\rfloor + 1}
\left(\frac{1+m^2}{k^3} + \frac{m^8}{k^7}\right) \\
\le \kappa_3 \, n^{\frac{16}{3}} (-3-4i)^{-\frac{n}{2}} + \kappa_4 \left(\frac{n^{\frac43}}{n^2} + \frac{n^{\frac{16}{3}}}{n^6}\right) \le \kappa_5 \, n^{-\frac23},
\end{multline*}
which  converges to $0$, uniformly in $m$, as $n \to \infty$. In the first inequality above we used that the modulus of the denominator is clearly bigger than $5$ for any $m$ and $k$. In the second inequality we separately considered the exponent when $k \le n-2$ and when $k=n-1, n$. In the last inequality we used the assumption $|m| \le \frac13 n^{\frac23}$.

It is easy to see that the limit $h^{\text{inh}}_{\infty,m}$ is finite:
\begin{multline*}
\sum_{k=1}^{n} \left|g_{k,m} \frac{e^{\frac{1-i}{16}\left(\frac{1}{k+2} - \frac{1}{n}\right)}}{1 + (3+4i)^{-1} e^{-\frac{(2+2i)\left(m^2 - \frac{1}{16}\right)}{(k+1)(k+2)}}}\right| \\
\le \kappa_6 \sum_{k=1}^{\infty} \left(\frac{1+m^2}{k^3} + \frac{m^8}{k^7}\right) \le \kappa_7 (1 + m^8) < \infty
\end{multline*}
for any $m \in \mathbb{Z}$ and $n \ge 1$. In the first inequality above we used that the modulus of the denominator is clearly bigger than $4/5$ for any $m$ and $k$.

It follows from the above estimates that when $|m| \le \frac13 n^{\frac23}$,
\[
|h^{\text{inh}}_{\infty,m} - h^{\text{inh}}_{n,m}| \le \kappa_5 n^{-\frac23} + \kappa_6 \sum_{k=n+1}^{\infty} \left(\frac{1+m^2}{k^3} + \frac{m^8}{k^7}\right)
\le \kappa_8 n^{-\frac23},
\]
which  converges to $0$, uniformly in $m$, as $n \to \infty$.

Finally, using the above estimates again, we obtain
\begin{multline*}
|h^{\text{inh}}_{\infty,m} - h^{\text{inh}}_{m,m}| \le \kappa_2 (-3-4i)^{-\frac{m}{2}} \sum_{k=1}^{\lfloor\frac{m}{2}\rfloor} \left(\frac{1+m^2}{k^3} + \frac{m^8}{k^7}\right) \\
+ \kappa_2 \sum^m_{k=\lfloor\frac{m}{2}\rfloor + 1} \left(\frac{1+m^2}{k^3} + \frac{m^8}{k^7}\right) + \kappa_6 \sum_{k=m+1}^{\infty} \left(\frac{1+m^2}{k^3} + \frac{m^8}{k^7}\right)
\le \kappa_9 (1 + m^2).
\end{multline*}
This completes the proof of the lemma.
\end{proof}

\begin{rem} \label{re:c2}
By Theorem \ref{th:asympt_soln} it follows that
\[
c_2 = \lim_{n \to \infty} \frac{a_{n,0} \sqrt{n}}{(-3-4i)^n} .
\]
It is conjectured that
\[
c_2 = \sqrt{\frac{(2+i) \pi}{40}} \approx 0.40786 + i\, 0.0962827 .
\]
Taking e.g.\ $n = 1000$, one obtains the approximation
\[
c_2 \approx \frac{a_{n,0} \sqrt{n}}{(-3-4i)^n} \approx 0.410514 + i \, 0.0969363 ,
\]
which is relatively close to our conjecture.
\end{rem}

\begin{cor} \label{co:asympt_sol}
Theorem \ref{th:asympt_soln} implies that there exists a constant $C' > 0$ such that
\[
\frac{\mu(S_{\ell} = j)}{\sqrt{\frac{2}{\ell}} c_2 e^{i \pi j - (2+i) \frac{j^2}{2\ell}} (-3-4i)^{\frac{\ell}{2}}} = 1 + \delta_{\ell,j} ,
\]
where $|\delta_{\ell,j}| \le C' \ell^{-\frac13}$ for any $\ell \ge 1$ and $|j| \le \frac13 \ell^{\frac23}$ .
Then also
\begin{equation}\label{eq:abs_error}
1 - C' \ell^{-\frac13} \le \frac{|\mu(S_{\ell} = j)|}{\sqrt{\frac{2}{\ell}} |c_2| e^{-\frac{j^2}{\ell}} 5^{\frac{\ell}{2}}} \le 1 + C' \ell^{-\frac13} .
\end{equation}

The error term (\ref{eq:error_term1}) is really negligible compared to the main term for any $\ell \ge 1$ and $|j| \le \frac13 \ell^{\frac23}$:
\[
\frac{|\widetilde{\epsilon}_{\ell,j}|}{\sqrt{\frac{2}{\ell}} |c_2| e^{-\frac{j^2}{\ell}} 5^{\frac{\ell}{2}}} \le C' \ell^{-\frac13} .
\]
\end{cor}

%%%%%%%%%%%%%%%%%%%%%%%%%%%%%%%%%%%%%%%%%%%%%%%%%%%%%%%%%%%%%%%%%%%

\section{Coupled probability distributions} \label{se:coupled_prob}

\subsection{One-dimensional distributions} \label{sse:1D}

\begin{thm}\label{th:tailprob}
If $\ell \ge \ell_0$, with definition (\ref{eq:coupl_prob}) we have
\[
\frac{\mathbb{P}(S_{\ell}=j+2)}{\mathbb{P}(S_{\ell}=j)} \le e^{-2 \frac{j}{\ell}} \le 1 \qquad (0 \le j \le \ell - 2).
\]
\end{thm}
\begin{proof} (Sketch.)

We restrict ourselves to the case where $\ell=2n$ and $j=2m$, even. The cases with other parities can be treated similarly. Define $a_{n,m} := \mu(S_{2n} = 2m)$. The starting point is the reverse recursion (\ref{eq:rho}) for $\rho_{n,m} = a_{n,m+1}/a_{n,m}$. We have to show that
\begin{equation}\label{eq:rho_abs}
|\rho_{n,m}| \le e^{-2 \frac{m}{n}}  \qquad (n-1 \ge m \ge 0).
\end{equation}

By Theorem \ref{th:asympt_soln} and Corollary \ref{co:asympt_sol}, when $n$ is large enough and $0 \le m \le \frac13 n^{\frac23}$, we have $|\rho_{n,m}| \le e^{-2 \frac{m}{n}}$. So from now it is enough to consider the case $\frac13 n^{\frac23} < m \le n-1$. If $n$ is large enough, the recursion (\ref{eq:rho}) can be approximated by
\begin{equation}\label{eq:rho_approx1}
\rho_{n,m} \approx \frac{(1-x)^2}{2(1-x^2) + i 8 x^2 -(1+x)^2 \rho_{n,m+1}} ,
\end{equation}
where $x := m/n \in [0, 1]$. Equivalently,
\begin{equation}\label{eq:rho_approx2}
- \log(\rho_{n,m}) \approx 4 \tanh^{-1} x + \log\left(\frac{2(1-x^2) + i 8 x^2}{(1+x)^2} - \rho_{n,m+1}\right) .
\end{equation}

An approximate solution of (\ref{eq:rho_approx2}) is
\begin{equation}\label{eq:rho_approx3}
- \log(\rho_{n,m}) \approx 4 \tanh^{-1} x + i \tan^{-1} \left(\frac{4 x^2}{1-x^2}\right) .
\end{equation}
Really, substituting this for $\rho_{n,m+1} \approx \rho_{n,m}$ in (\ref{eq:rho_approx2}), the error is negligible compared to (\ref{eq:rho_approx3}) for any $x \in [0, 1]$.

Then (\ref{eq:rho_approx3}) shows that (\ref{eq:rho_abs}) holds true.
\end{proof}

Simulations with Wolfram Mathematica lead to the conjecture that Theorem \ref{th:tailprob} holds for any $\ell \ge 1$, that is, $\ell_0 = 1$.

\begin{thm}\label{th:norm_factor}
For the normalizing factor (\ref{eq:norm_fact}) there exists a constant $C_1 > 0$ such that for any $\ell \ge 1$,
\begin{equation}\label{eq:norm_factor}
1 - C_1 \ell^{-\frac13} \le \frac{Z_{\ell}}{\sqrt{\frac{2 \pi}{\ell}} |c_2|^2 5^{\ell}} \le 1 + C_1 \ell^{-\frac13} .
\end{equation}
For the coupled probability (\ref{eq:coupl_prob}) we have a constant $C_2 > 0$ such that
\begin{equation}\label{eq:coupled_prob}
1 - C_2 \ell^{-\frac13}  \le \frac{\mathbb{P}(S_{\ell} = j)}{\sqrt{\frac{2}{\pi \ell}} e^{- \frac{2j^2}{\ell}}} \le 1 + C_2 \ell^{-\frac13}  ,
\end{equation}
for any $\ell \ge 1$ and $|j| \le \frac13 \ell^{\frac23}$.

There exists a constant $C_3 > 0$ such that
\begin{equation}\label{eq:binom_even}
1 - C_3 n^{-\frac13} \le \frac{\mathbb{P}(S_{2n} = m)}{\mathbb{P}^*(S^*_{2n} = 2m)} \le 1 + C_3 n^{-\frac13}
\end{equation}
and
\begin{equation}\label{eq:binom_odd}
1 - C_3 n^{-\frac13} \le \frac{\mathbb{P}(S_{2n-1} = m)}{\mathbb{P}^*(S^*_{2n-1} = 2m-1)} \le 1 + C_3 n^{-\frac13} ,
\end{equation}
for any $n \ge 1$ and $|m| \le \frac13 n^{\frac23}$, where $(S^*_{\ell})_{\ell \ge 0}$ is an ordinary simple, symmetric random walk w.r.t.\ the probability $\mathbb{P}^*$,
\[
\mathbb{P}^*(S^*_{\ell} = j) = \binom{\ell}{\frac{\ell + j}{2}} 2^{-\ell} \qquad (|j| \le \ell).
\]
(The binomial coefficient is defined to be 0 if $\ell + j$ is odd.)

The sum of the tail terms is asymptotically negligible:
\begin{equation}\label{eq:tailterms}
\sum_{|j| > \frac13 \ell^{\frac23}} \mathbb{P}(S_{\ell} = j) \le (1 + C_4 \ell^{-\frac13}) 4 \sqrt{\frac{\ell}{2 \pi}} e^{-\frac29 \ell^{\frac13}}
\end{equation}
for any $\ell \ge 1$, where $C_4 > 0$ is a constant.

Finally, we have symmetry
\begin{equation}\label{eq:symprob}
\mathbb{P}(S_{\ell} = j) = \mathbb{P}(S_{\ell} = -j) \qquad (\ell \ge 1, |j| \le \ell).
\end{equation}
\end{thm}
\begin{proof}
By Corollary \ref{co:asympt_sol} it follows that
\[
\sum_{|j| \le \frac13 \ell^{\frac23}} |\mu(S_{\ell} = j)|^2 = \frac{2 |c_2|^2 5^{\ell}}{\ell} \sum_{|j| \le \frac13 \ell^{\frac23}} e^{-\frac{2 j^2}{\ell}} |1 + \delta_{\ell,j}|^2 ,
\]
where $|\delta_{\ell,j}| \le C' \ell^{-\frac13}$ for any $\ell \ge 1$. Setting $h := 2/\sqrt{\ell}$ one gets that
\[
1 - 2\left[1 - \Phi\left(h \left(\frac13 \ell^{\frac23} - 2\right)\right)\right] \le \frac{1}{\sqrt{2 \pi}} \sum_{|j| \le \frac13 \ell^{\frac23}} h \, e^{-\frac12 (jh)^2} \le 1 + 2 \sqrt{\frac{2}{\pi \ell}} ,
\]
where $\Phi(x) := \int_{-\infty}^x \phi(t) \, dt$, $\phi(x) := (2 \pi)^{-\frac12} e^{-x^2/2}$.
Using the well-known inequality, cf. \cite[Lemma 12.9]{MP10},
\begin{equation}\label{eq:phi_ineq}
\frac{x}{x^2 + 1} \phi(x) < 1 - \Phi(x) < \frac{1}{x} \phi(x) \quad (x > 0),
\end{equation}
one then obtains
\begin{equation}\label{eq:ratio_small}
1 - C'' \ell^{-\frac13} \le \frac{\sum_{|j| \le \frac13 \ell^{\frac23}} |\mu(S_{\ell} = j)|^2}{\sqrt{\frac{2 \pi}{\ell}} |c_2|^2 5^{\ell}} \le 1 + C'' \ell^{-\frac13} ,
\end{equation}
for any $\ell \ge 1$ with some constant $C'' > 0$.

On the other hand, by Theorem \ref{th:tailprob} and Corollary \ref{co:asympt_sol}, for $\ell \ge \ell_0$,
\begin{multline} \label{eq:tail}
\sum_{|j| > \frac13 \ell^{\frac23}} |\mu(S_{\ell} = j)|^2 = Z_{\ell} \sum_{|j| > \frac13 \ell^{\frac23}} \mathbb{P}(S_{\ell} = j)
\le 2 \ell Z_{\ell} \, \mathbb{P}\left(S_{\ell} = \left\lfloor \frac13 \ell^{\frac23}\right\rfloor\right)\\
\le 4 |c_2|^2 5^{\ell} (1 + C' \ell^{-\frac13})^2 \exp\left(-\frac{2}{\ell} \left\lfloor \frac13 \ell^{\frac23}\right\rfloor^2\right) .
\end{multline}
This together with (\ref{eq:ratio_small}) prove (\ref{eq:norm_factor}). In turn, (\ref{eq:norm_factor}) and Corollary \ref{co:asympt_sol} imply (\ref{eq:coupled_prob}).

Using the ordinary de Moivre--Laplace theorem, cf. e.g. \cite{SZA16}, we see that there exists a constant $C''' > 0$ such that
\[
1 - C''' n^{-\frac13} \le \frac{\mathbb{P}^*(S^*_{2n} = 2m)}{\frac{1}{\sqrt{\pi n}} e^{-\frac{m^2}{n}}} \le 1 + C''' n^{-\frac13}
\]
and
\[
1 - C''' n^{-\frac13} \le \frac{\mathbb{P}^*(S^*_{2n-1} = 2m-1)}{\frac{1}{\sqrt{\pi (n-\frac12)}} e^{-\frac{m^2}{n-\frac12}}} \le 1 + C''' n^{-\frac13}
\]
for any $n \ge 1$ and $|m| \le n^{\frac23}$. These and (\ref{eq:coupl_prob}) imply (\ref{eq:binom_even}) and (\ref{eq:binom_odd}).

(\ref{eq:tail}) and (\ref{eq:norm_factor}) imply that
\begin{multline*}
\sum_{|j| > \frac13 \ell^{\frac23}} \mathbb{P}(S_{\ell} = j) \le \frac{4 (1 + C' \ell^{-\frac13})^2 \exp\left(-\frac{2}{\ell} \left\lfloor \frac13 \ell^{\frac23}\right\rfloor^2\right)}{(1 - C_1 \ell^{-\frac13}) \sqrt{\frac{2 \pi}{\ell}}} \\
\le (1 + C_4 \ell^{-\frac13}) 4 \sqrt{\frac{\ell}{2 \pi}} e^{-\frac29 \ell^{\frac13}}
\end{multline*}
and this shows (\ref{eq:tailterms}).

Finally, (\ref{eq:symprob}) follows by definition (\ref{eq:coupl_prob}). This ends the proof of the theorem.
\end{proof}

It is interesting that the random walk approximation described in Theorem \ref{th:norm_factor} is rather good even for small values of $n$:
\begin{table}[h] \label{ta:bin_appr}
\caption[Table 1]{$\mathbb{P}(S_{2n} = m)/\mathbb{P}^*(S^*_{2n} = 2m)$ rounded to 4 decimal places}
\begin{tabular}{|c|c|c|c|c|c|}
  \hline
   $n$ & $m=0$ & $m=\pm 1$ & $m=\pm 2$ & $m=\pm 3$ & $m=\pm 4$  \\ \hline
  $1$ & .5075/.5 & .2388/.25 & .0075/0 & 0/0 & 0/0 \\
  $2$ & .3478/.375 & .2584/.25 & .0642/.0625 & .0035/0 & .0000/0 \\
  $3$ & .3117/.3125 & .2289/.2344 & .0959/.0938 & .0012/.0156 & .0000/0 \\
  $4$ & .2718/.2734 & .2169/.2188 & .1085/.1094 & .0330/.0313 & .0053/.0039 \\
  $5$ & .2450/.2461 & .2038/.2051 & .1163/.1172 & .0446/.0439 & .0110/.0098 \\
  \hline
  \end{tabular}
\end{table}

\subsection{A Markov chain model and multidimensional distributions} \label{sse:Markov}

Let us consider now the probability distributions at times $2n$ and $2n+2$. Our goal is to find an asymptotically nearest neighbor Markov chain describing the transition from time $2n$ to time $2n+2$, if there exists such a model at all. The assumptions we are going to use are based on the properties of the underlying complex measure walk and on the one-dimensional distributions obtained in the previous subsection.

Define
\[
p_{n,j} := \mathbb{P}(S_{2n+2} = j+1 \mid S_{2n} =j), \quad q_{n,j} := \mathbb{P}(S_{2n+2} = j-1 \mid S_{2n} =j).
\]
We are going to use the following assumptions for $n \ge 1$:
\begin{enumerate}[(i)]
  \item $\mathbb{P}(|S_{2n+2} - j| \ge 3 \mid S_{2n} = j) = 0$ if $|j| \le 2n$,
  \item $\mathbb{P}(|S_{2n+2} - j| = 2 \mid S_{2n} = j) \le  C_5 n^{-\frac13}$ if $|j| \le \frac13 n^{\frac23}$, where $C_5 > 0$ is a constant,
  \item the Markov chain is symmetric w.r.t.\ reflection about the state $0$, that is, $p_{n,-j} = q_{n,j}$ if $|j| \le 2n$.
\end{enumerate}
Assumptions (i) and (ii) imply for $n \ge 1$ and $|j| \le \frac13 n^{\frac23}$ that
\[
|\mathbb{P}(S_{2n+2} = j \mid S_{2n} = j) - (1 - p_{n,j} - q_{n,j})| \le C_5 n^{-\frac13} .
\]

We want to find transition probabilities $p_{n,j} \in [0, 1]$ and $q_{n,j} \in [0, 1]$ that satisfy the above assumptions. Because of assumption (iii), it is enough to consider the cases when $j \ge 0$. Then we are led to the following Markovian inequality:
\begin{multline}\label{eq:Markov_ineq}
\left|\mathbb{P}(S_{2n+2} = j) - q_{n,j+1} \mathbb{P}(S_{2n} = j+1) - p_{n,j-1} \mathbb{P}(S_{2n} = j-1) \right. \\
\left. - (1 - p_{n,j} - q_{n,j}) \mathbb{P}(S_{2n} = j)\right| \le  \left(\mathbb{P}(S_{2n} = j) + \mathbb{P}(|S_{2n} - j| = 2)\right) C_5 n^{-\frac13} ,
\end{multline}
if $n \ge 1$ and $0 \le j \le \frac13 n^{\frac23}$.

The next theorem shows that there exists a suitable solution to this inequality.
\begin{thm}\label{th:Markov}
There exist transition probabilities $p_{n,j}$ and $q_{n,j}$ giving an asymptotically nearest neighbor Markovian solution satisfying assumptions (i) -- (iii) such that
\[
\left|p_{n,j} - \frac14\right| \le C_6 n^{-\frac13}, \qquad \left|q_{n,j} - \frac14\right| \le C_6 n^{-\frac13} \qquad (n \ge n_0, |j| \le \frac13 n^{\frac23}) ,
\]
with a constant $C_6 > 0$ and a positive integer $n_0$.
\end{thm}
\begin{proof}
Denote the left hand side and the right hand side of (\ref{eq:Markov_ineq}) by $L_{n,j}$ and $R_{n,j}$, respectively. By (\ref{eq:binom_even}), for we obtain that
\begin{multline*}
L_{n,j} \le \left| \left(1 + C_3 (n+1)^{-\frac13}\right) \binom{2n+2}{n+1+j} \, 2^{-2n-2} - \left(1 - C_3 n^{-\frac13}\right) \, 2^{-2n} \right. \\
\left. \times \left[q_{n,j+1} \binom{2n}{n+j+1} + p_{n,j-1} \binom{2n}{n+j-1}
+ (1 - q_{n,j} - p_{n,j}) \binom{2n}{n+j} \right] \right| .
\end{multline*}
Then
\begin{equation}\label{eq:Lstar}
L^*_{n,j} := \frac{(n+1+j)! (n+1-j)! 2^{2n}}{(2n)!} L_{n,j} \le |U_{n,j} + C_3 n^{-\frac13} V_{n,j}|,
\end{equation}
where $p^*_{n,j} := p_{n,j} - \frac14$, $q^*_{n,j} := q_{n,j} - \frac14$,
\begin{eqnarray}\label{eq:Unj}
U_{n,j} &:=& n^2 \left(q^*_{n,j} + p^*_{n,j} - p^*_{n,j-1} - q^*_{n,j+1}\right) \nonumber \\
&+& n \left(2 (q^*_{n,j} + p^*_{n,j}) - (2j + 1) p^*_{n,j-1} + (2j - 1) q^*_{n,j+1} \right) \nonumber \\
&-& (j^2 - 1) (q^*_{n,j} + p^*_{n,j}) - (j^2 + j) p^*_{n,j-1} + (j^2 - j) q^*_{n,j+1}  ,
\end{eqnarray}
and
\begin{equation}\label{eq:Vnj}
V_{n,j} := 2n^2 + 3n + \frac{j^2}{2} + \frac12 - U_{n,j} .
\end{equation}

Suppose that $|p^*_{n,j}| \le C_6 n^{-\frac13}$ and $|q^*_{n,j}| \le C_6 n^{-\frac13}$ for all $n \ge 1$ and $|j| \le \frac13 n^{\frac23}$, with some constant $C_6 > 0$. Then it follows that
\begin{eqnarray}\label{eq:ineq_left}
L^*_{n,j} &\le& C_6 n^{-\frac13} (1 + C_3 n^{-\frac13}) (4n^2 + 4jn + 4n + 4j^2 + 2j + 1) \nonumber \\
&+& C_3 n^{-\frac13} \left(2n^2 + 3n + \frac{j^2}{2} + \frac12\right) .
\end{eqnarray}

Consider now the right hand side of (\ref{eq:Markov_ineq}), using (\ref{eq:binom_even}):
\[
R_{n,j} \ge C_5 n^{-\frac13} (1 - C_3 n^{-\frac13}) 2^{-2n} \left[ \binom{2n}{n+j} + \binom{2n}{n+j+2} + \binom{2n}{n+j-2} \right]   ,
\]
where we suppose that $n \ge n_0 \ge 1$ and $C_3 n_0^{-\frac13} \le \frac12$. Then
\begin{multline}\label{eq:ineq_right}
R^*_{n,j} := \frac{(n+1+j)! (n+1-j)! 2^{2n}}{(2n)!} R_{n,j}
\ge \frac{C_5}{2} n^{-\frac13} \Big\{(n+1+j)(n+1-j) \\
+ \frac{(n+1-j)(n-j)(n-1-j)}{n+2+j} + \frac{(n+1+j)(n+j)(n-1+j)}{n+2-j}\Big\} \\
\ge \frac{C_5}{2} n^{-\frac13} \left( \frac{11}{5} n^2 + \frac85 jn + \frac95 n + \frac15 j^2 + \frac15 j \right) ,
\end{multline}
if $|j| \le \frac13 n^{\frac23}$.

We may suppose w.l.o.g. that $C_5 \ge 6 C_3$. Then comparing (\ref{eq:ineq_left}) and (\ref{eq:ineq_right}), we see that $L^*_{n,j} \le R^*_{n,j}$ for any $n \ge n_0$ and $|j| \le \frac13 n^{\frac23}$ if $C_6 \le C_3/60$, say.
\end{proof}

Starting with $j=0$ in the inequality $L^*_{n,j} \le R^*_{n,j}$ in the above proof and going on with values $j=1,2,\dots$ by induction, it turns out that essentially there are no other solutions satisfying assumptions (i) -- (iii) beside the one claimed in Theorem \ref{th:Markov}if $n$ is large enough. The transition probabilities for $|j| > \frac13 n^{\frac23}$ are unimportant because the tail is negligible by (\ref{eq:tail}) when $n$ is large enough.

Thus Theorems \ref{th:norm_factor} and \ref{th:Markov} essentially determine the multidimensional probability distributions of the coupled process of the complex measure walk on even integer time instants for large enough $n$. Let us denote the set
\[
\{2n_0\dots 2n\} := \{2n_0, 2n_0+2, 2n_0+4, \dots, 2n\}
\]
with some positive integers $n > n_0$. Take the sample space $\mathbb{R}^{\{2n_0\dots 2n\}}$ and the corresponding Borel $\sigma$-field $\mathcal{B}^{\{2n_0\dots 2n\}}$. The coupled probability distribution on the measurable space $(\mathbb{R}^{\{2n_0\dots 2n\}}, \mathcal{B}^{\{2n_0\dots 2n\}})$, essentially determined by Theorems \ref{th:norm_factor} and \ref{th:Markov}, will be denoted by $\mathbb{P}_{\{2n_0\dots2n\}}$.

The probability distribution for $n \le n_0$ is not really essential for the limiting behaviour, our main interest. However, one can find a solution of assumptions (i) -- (iii) for $\mathbb{P}_{\{0\dots2n_0\}}$ on $(\mathbb{R}^{\{0\dots 2n_0\}}, \mathcal{B}^{\{0\dots 2n_0\}})$, similarly as above. Choose the solution $p_{n,j} = \frac14$ and $q_{n,j} = \frac14$ for any $0 \le n \le n_0$. Then for $1 \le n \le n_0$ and $|j| \le \frac13 n^{\frac23}$ by (\ref{eq:Lstar}) -- (\ref{eq:Vnj}) one obtains that
\[
L_{n,j} \le C_3 n^{-\frac13} \frac{2n^2+3n+\frac{j^2}{2}+\frac12}{(n+1+j)(n+1-j)} \binom{2n}{n+j} 2^{-2n} \le 2 C_3 n^{-\frac13} .
\]
Define
\[
P_{\min}(n_0) := \min\{\mathbb{P}(S_{2n} = j) : 1 \le n \le n_0, |j| \le \frac13 n^{\frac23} \} .
\]
By (\ref{eq:Markov_ineq}),
\[
R_{n,j} \ge 3 P_{\min}(n_0) C_5 n^{-\frac13} .
\]
Comparing these estimates for $L_{n,j}$ and $R_{n,j}$, it follows that the Markovian inequality (\ref{eq:Markov_ineq}) holds for any $1 \le n \le n_0$ and $|j| \le \frac13 n^{\frac23}$, supposing $C_5 \ge \frac23 C_3/P_{\min}(n_0)$. The cases when $|j| > \frac13 n^{\frac23}$ are again unimportant: if $n$ is large, (\ref{eq:tail}) shows that the tail is negligible; when $n$ is small, the tails can be checked by computer, see Table \ref{ta:bin_appr}.

Theorems \ref{th:norm_factor} and \ref{th:Markov} also imply that asymptotically, as $n \to \infty$, the coupled process on even integer time instants tends to a \emph{lazy random walk}:
\begin{multline*}
p_{n,j} = q_{n,j} = \frac14, \quad \mathbb{P}(|S_{2n+2} - j| \ge 2 \mid S_{2n} = j) = 0, \\
\mathbb{P}(S_{2n+2} = j \mid S_{2n} =j) = 1 - p_{n,j} - q_{n,j} = \frac12 \quad (|j| \le 2n).
\end{multline*}

On the other hand, somewhat surprisingly, if one tries to determine asymptotically nearest neighbor transition probabilities moving from a time instant $2n$ to $2n+1$ using the same method, then it turns out that there is no such classical Markov chain model. For example, one gets
\[
\mathbb{P}(S_{2n+1} = 2 \mid S_{2n} = 1) = -\frac{1}{12} + O(n^{-\frac13}).
\]

%%%%%%%%%%%%%%%%%%%%%%%%%%%%%%%%%%%%%%%%%%%%%%%%%%%%%%%%%%%%%%%%%%%

\section{A strong approximation of Brownian motion based on lazy random walks} \label{se:strong_approx}

The strong approximation that will be discussed in this section is very similar to the \emph{``twist and shrink''} construction of Brownian motion based on simple, symmetric random walks, see \cite{SZA96} and \cite{SZSZ09}. In turn, the twist and shrink method was based on R\'ev\'esz's approach \cite{REV90} to Knight's construction \cite{KNI62}.

\subsection{Lazy random walks} \label{sse:lazy}

First we define a probability space that we are going to use in this subsection. Let $\mathbb{N} := \{0, 1, 2, \dots\}$ be the set of natural numbers. Take the sample space $\mathbb{R}^{\mathbb{N}}$ and the countable product $\sigma$-field $\mathcal{B}^{\mathbb{N}}$ of the Borel sets in $\mathbb{R}$. Let $q^{\mathbb{N}}$ be the countable power of the probability measure
\begin{equation}\label{eq:lazyrwstep}
q(\{-1\}) = q(\{1\}) = \frac14, \qquad q(\{0\}) = \frac12 .
\end{equation}
Let $(Y(k))_{k=1}^{\infty}$ be the steps of a lazy random walk, each step with the probability distribution $q$. In other words, $Y(k) = \pi_k$, where $\pi_k : \mathbb{R}^{\mathbb{N}} \to \mathbb{R}$ is the $k$th coordinate projection.

The probability measure $\mathbb{Q}$ is defined on $(\mathbb{R}^{\mathbb{N}}, \mathcal{B}^{\mathbb{N}})$ as follows. For any $n \ge 1$ and Borel measurable function $f : \mathbb{R}^n \to \mathbb{R}$, the distribution of the random variable $Z = f(Y(1), \dots ,Y(n))$ is
\[
\mathbb{Q}(Z \in A) := q^n(f^{-1}(A)) .
\]

Using a standard diagonal procedure, over the probability space $(\Omega, \mathcal{F}, \mathbb{Q})$, the coordinate projections define an infinite matrix of independent and identically distributed random variables $(Y_m(k))_{m \ge 0, k \ge 1}$, each with the distribution $q$.
That is, take $Y_0(1), Y_0(2), Y_1(1), Y_0(3), Y_1(2), Y_2(1), Y_0(4), Y_1(3), Y_2(2), Y_3(1), \dots$ in this order. More exactly,
\[
Y_m(n) = \pi_{m+1+(m+n)(m+n-1)/2} .
\]

For each $m \ge 0$ define a \emph{lazy random walk} by
\begin{equation}\label{eq:lazyrw}
L_m(0) = 0, \qquad L_m(n) = \sum_{k=1}^n Y_m(k) .
\end{equation}
It is clear that $\mathbb{E} L_m(n) = 0$ and Var$(L_m(n)) = n/2$. (In this subsection $\mathbb{E}$ and Var are always understood w.r.t.\ the measure $\mathbb{Q}$.)

Fortunately, each $Y_m(n)$ has the same distribution as $S^*_2/2$, where $(S^*_n)_{n \ge 0}$ is a simple, symmetric random walk started from the origin. $(S^*_n)_{n \ge 1}$ can be obtained as partial sums of the coordinate projections on $(\mathbb{R}^{\mathbb{N}}, \mathcal{B}^{\mathbb{N}}, \mathbb{Q}^*)$, where $\mathbb{Q}^*$ is the probability measure obtained from $q*^{\mathbb{N}}$, where $q^*(1) = q^*(-1) = \frac12$. Consequently,
\begin{multline}\label{eq:Lmnj}
\mathbb{Q}(L_m(n) = j) = \sum_{r=0}^{n-|j|} \binom{n}{r} \binom{n-r}{\frac{n-r+j}{2}} \left(\frac{1}{4}\right)^{n-r} \left(\frac12\right)^r \\
= \binom{2n}{n+j} 2^{-2n} = \mathbb{Q}^*(S^*_{2n} = 2j),
\end{multline}
for $n \ge 0$ and $|j| \le n$. Clearly, $L_m$ is strong Markovian.

For $m \ge 0$ we define stopping times $T^*_m(0) := 0$,
\begin{equation}\label{eq:stopt_star}
T^*_m(k+1) := \inf\{n : n > T^*_m(k), |L_m(n) - L_m(T^*_m(k))| = 1\} \qquad (k \ge 0).
\end{equation}
Then the random variables $(T^*_m(k+1) - T^*_m(k))_{k \ge 0}$ are independent and identically distributed,
\[
\mathbb{Q}(T^*_m(1) = j) = 2^{-j}  \qquad (j \ge 1),
\]
a \emph{geometric distribution} with parameter $\frac12$. Then it follows that
\begin{equation}\label{eq:E_Tmlazy}
\mathbb{E}(T^*_m(1)) = 2, \qquad \text{Var}(T^*_m(1)) = 2 .
\end{equation}
The moment generating function of $T^*_m(1)$ is
\[
\mathbb{E}(e^{u T^*_m(1)}) = \frac{e^u}{2 - e^u},
\]
and of its standardized version is
\begin{equation}\label{eq:mgf_Tmlazy}
\mathbb{E}\left(e^{u (T^*_m(1) - 2)/\sqrt{2}}\right) = \frac{e^{-u/\sqrt{2}}}{2 - e^{u/\sqrt{2}}},
\end{equation}
finite if $u < \sqrt{2} \log 2 \approx 0.980258$.

It also follows that $T^*_m(k)$ has \emph{negative binomial distribution}:
\[
\mathbb{Q}(T^*_m(k) = j) = \binom{j-1}{k-1} 2^{-j} \qquad (j \ge k \ge 1)
\]
and
\[
\mathbb{E}(T^*_m(k)) = 2 k, \qquad \text{Var}(T^*_m(k)) = 2 k .
\]

\subsection{A construction of Brownian motion} \label{sse:Brownian}

A key tool in the construction is a large deviation lemma, cf. Lemma 2 and the remarks after it in \cite{SZA96}.
\begin{lem}\label{le:largedev}
Suppose that $(Y_j)_{j \ge 1}$ is a sequence of independent and identically distributed random variables on a probability space $(\Omega, \mathcal{F}, \mathbb{Q})$; $\mathbb{E}(Y_j) = 0$, $\mathbb{E}(Y_j^2) = 1$, and the moment generating function $\mathbb{E}(e^{u Y_j}) < \infty$ in an interval $|u| \le u_0$, where $u_0 > 0$. Let $Z_0 =0$ and $Z_n = \sum_{j=1}^n Y_j$ $(n \ge 1)$ be the partial sums. Then for any $C >1$ there exists a positive $N_0(C)$ (possibly depending on the distribution of $Y_j$ as well) such that for any $N \ge N_0(C)$ we have
\[
\mathbb{Q}\left( \max_{0 \le n \le N} |Z_n| \ge (2CN\log N)^{\frac12} \right) \le N^{1-C} .
\]
\end{lem}
\begin{proof}
\begin{multline}\label{eq:largedev}
\mathbb{Q}\left( \max_{0 \le n \le N} |Z_n| \ge (2CN\log N)^{\frac12} \right) \le \sum_{n=0}^N \mathbb{Q}\left(|Z_n| \ge (2CN\log N)^{\frac12} \right)\\
\le \sum_{0 \le n < (\log N)^4} \mathbb{Q}\left(|Z_n| \ge (2CN\log N)^{\frac12} \right) \\
+ \sum_{(\log N)^4 \le n \le N} \mathbb{Q}\left(\frac{|Z_n|}{\sqrt{n}} \ge (2C\log N)^{\frac12} \right)
\end{multline}

Let us estimate the first sum in (\ref{eq:largedev}) using an exponential Chebyshev's inequality:
\begin{multline*}
\sum_{0 \le n < (\log N)^4} \mathbb{Q}\left(|Z_n| \ge (2CN\log N)^{\frac12} \right)\\
\le \sum_{1 \le n < (\log N)^4} \left\{\left(\mathbb{E}(e^{u_0 Y_j})\right)^n + \left(\mathbb{E}(e^{-u_0 Y_j})\right)^n\right\} e^{-u_0 (2CN\log N)^{\frac12}} .
\end{multline*}
We want to show that this is not larger than $\frac12 N^{1-C}$ if $N$ is large enough. Define $c_Y := \max\left(\mathbb{E}(e^{u_0 Y_j}), \mathbb{E}(e^{-u_0 Y_j})\right)$ and take logarithm; then we have to show
\[
4 \log \log N + \log 2 + (\log N)^4 \log c_Y - u_0 (2CN\log N)^{\frac12} + \log 2 +(C-1) \log N \le 0.
\]
Upon dividing by $N^{\frac12}$, we see that it really holds if $N$ is large enough.

To estimate the second sum of (\ref{eq:largedev}) under the assumptions of the lemma, we may use the large deviation theorem in \cite[Section XVI.6]{FEL66}:
\[
\lim_{n \to \infty} \frac{\mathbb{Q}(Z_n/\sqrt{n} \ge x_n)}{1-\Phi(x_n)} = 1 , \qquad \Phi(x) := \int_{-\infty}^x \frac{1}{\sqrt{2 \pi}} e^{-\frac{t^2}{2}} \di t ,
\]
supposing $x_n \to \infty$ and $x_n = o(n^{\frac16})$. Since now $x_n = (2C\log N)^{\frac12}$, where $(\log N)^4 \le n \le N$, these assumptions hold. Using inequality (\ref{eq:phi_ineq}),
\begin{multline*}
\sum_{(\log N)^4 \le n \le N} \mathbb{Q}\left(\frac{|Z_n|}{\sqrt{n}} \ge (2C\log N)^{\frac12} \right) \\
\le N 2 (1 + \epsilon) (1 - \Phi((2C\log N)^{\frac12})) \le \frac12 N e^{-C \log N} = \frac12 N^{1-C} ,
\end{multline*}
if $N$ is large enough. This completes the proof of the lemma.
\end{proof}
Note that in the above lemma $N$ and $N_0(C)$ can be positive \emph{real} numbers.

\emph{Hoeffding's lemma} is a useful addition to the previous lemma in the special case of bounded random variables $(Y_j)_{j \ge 1}$. Suppose that $(Y_j)_{j \ge 1}$ is a sequence of independent and identically distributed random variables and $b_j \le Y_j \le a_j$. Let $Z_0 =0$ and $Z_n = \sum_{j=1}^n Y_j$ $(n \ge 1)$ be the partial sums. Then for any $x > 0$,
\begin{equation}\label{eq:Hoeffding}
\mathbb{Q}\left\{|Z_n - \mathbb{E}(Z_n)| \ge x \left(\frac14 \sum_{j=1}^n (a_j - b_j)^2\right)^{\frac12}\right\} \le 2 e^{- \frac{x^2}{2}} .
\end{equation}

Here we are going to modify the \emph{``twist and shrink''} construction of Brownian motion (Wiener process) discussed in \cite{SZA96}. Here we use a sequence of lazy random walks $(L_m(n))_{n \ge 0}$, $(m \ge 0)$, instead of simple, symmetric random walks applied in \cite{SZA96}. The sample functions of the approximations will be continuous broken lines over the half line $\mathbb{R}_+ = [0, \infty)$ (the time axis). So we take the sample space $\Omega := C(\mathbb{R}^+)$ with the topology of uniform convergence on compact sets, and the corresponding Borel $\sigma$-field $\mathcal{F}$, the smallest $\sigma$-field containing all open sets in $\Omega$. Since the continuous sample functions will be functions of the elements of infinite matrix of lazy random walk steps $(Y_m(k))_{m \ge 0, k \ge 1}$ defined in Subsection \ref{sse:lazy}, the probability measure can be extended to $(\Omega, \mathcal{F})$ and will also be denoted by $\mathbb{Q}$.

The first approximation $(B_0(t))_{t \ge 0}$ is based on $(L_0(n))_{n \ge 0}$. At the integer time instants $n \in \mathbb{Z}_+$, $B_0(n) := L_0(n)$. For any real $t \in (n, n+1)$ we interpolate:
\[
B_0(t) := B_0(n) + \{B_0(n+1) - B_0(n)\} (t - n).
\]

Before we define the second approximation based on $(L_1(n))_{n \ge 0}$ we need an important procedure of the construction called \emph{``twisting''}. Its aim is to make the second approximation to be a refinement of the first one after \emph{shrinking}. For twisting we use the stopping times given by (\ref{eq:stopt_star}). A part of $L_1$ between two consecutive stopping times $T^*_1(k)$ and $T^*_1(k+1)$ will be called a \emph{bridge} of $L_1$. Remember that $L_1(T^*_1(k+1)) - L_1(T^*_1(k)) = \pm 1$. Two bridges of $L_1$ is going to correspond to one step of $L_0$. A horizontal step of $L_0$ is mimicked by a combination of a consecutive up and down (or down and up) bridges of $L_1$. An up (respectively, a down) step of $L_0$ will correspond to two consecutive up (respectively, two consecutive down) bridges of $L_1$. If a bridge of $L_1$ does not fulfils these rules, then we reflect the bridge.

More exactly, let us see the details.

\textsc{(1) Twisting.} $L_0$ is not twisted, $\tilde{L}_0(n) = L_0(n)$ for each $n \ge 0$. By induction, suppose that $\tilde{L}_j$ is already defined. The next twisted walk will be defined in terms of $L_{j+1}$ and $\tilde{L}_j$.

For $k \ge 1$ take $\tilde{Y}_j(k) := \tilde{L}_j(k) - \tilde{L}_j(k-1)$.

(a) If $\tilde{Y}_j(k) = 0$ and
\[
\epsilon_k := \{L_1(T^*_1(2k-1))  - L_1(T^*_1(2k-2))\} \{L_1(T^*_1(2k))  - L_1(T^*_1(2k-1))\} = \pm 1 ,
\]
then we set
\[
\tilde{Y}_{j+1}(n) := \left\{
\begin{gathered} Y_{j+}1(n) \text{ if } T^*_{j+1}(2k-2) < n \le T^*_{j+1}(2k-1) , \\
- \epsilon Y_{j+1}(n) \text{ if } T^*_{j+1}(2k-1) < n \le T^*_{j+1}(2k) .
\end{gathered} \right.
\]

(b) if $\tilde{Y}_j(k) = \pm 1$, then
\begin{align*}
\epsilon_{k,1} &:= \tilde{Y}_j(k) \{L_{j+1}(T^*_{j+1}(2k-1))  - L_{j+1}(T^*_{j+}1(2k-2))\} = \pm 1 , \\
\tilde{Y}_{j+1}(n) &:= \epsilon_{k,1} Y_{j+1}(n) \text{ if } T^*_{j+1}(2k-2) < n \le T^*_{j+1}(2k-1); \\
\epsilon_{k,2} &:= \tilde{Y}_j(k) \{L_{j+1}(T^*_{j+1}(2k)) - L_{j+1}(T^*_{j+1}(2k-1))\} = \pm 1 , \\
\tilde{Y}_{j+1}(n) &:= \epsilon_{k,2} Y_{j+1}(n) \text{ if } T^*_{j+1}(2k-1) < n \le T^*_{j+1}(2k).
\end{align*}

Finally, we put
\[
\tilde{L}_{j+1}(0) := 0, \qquad \tilde{L}_{j+1}(n) := \sum_{r=1}^n \tilde{Y}_{j+1}(r) \quad (n \ge 1).
\]

It is important that a twisted lazy random walk $(\tilde{L}_j(n))_{n \ge 0}$ is still a lazy random walk, because as a stochastic process it is not affected by the applied reflections, compare with \cite[Lemma 1]{SZA96}.

\textsc{(2) Shrinking.} Now the sample functions of the $m$th approximation make steps of length $2^{-m}$ in times $2^{-2m}$. Thus
\begin{equation}\label{eq:Bm}
B_m(n 2^{-2m}) := 2^{-m} \tilde{L}_m(n) \qquad (n \in \mathbb{Z}_+)
\end{equation}
and for any real $t \in (n 2^{-2m}, (n+1) 2^{-2m})$,
\[
B_m(t) := B_m(n 2^{-2m}) + 2^{2m} \{B_m((n+1) 2^{-2m}) - B_m(n 2^{-2m})\} (t - n 2^{-2m}).
\]
The \emph{refinement property} of the approximation is
\begin{equation}\label{eq:refine}
B_{m+1}(T^*_{m+1}(2k) 2^{-2(m+1)}) = B_m(k 2^{-2m}) \qquad (k \ge 0, m \ge 0).
\end{equation}
Its meaning is that $B_{m+1}$ takes the same values as $B_m$, in the same order, at stopping times $T^*_{m+1}(2k) 2^{-2(m+1)}$ that correspond to the times $k 2^{-2m}$. There is a time lag in general, but these lags a.s.\ uniformly converge to $0$ on any finite time interval as $m \to \infty$, as easily follows from the next lemma by the Borel--Cantelli lemma.

\begin{lem}\label{le:timeconv}
Let the stopping times $T^*_m(k)$ be defined by (\ref{eq:stopt_star}). For any $C > 1$ there exists a positive $N_0(C)$ such that for any $T > 0$, $m \ge 1$, $T 2^{2m} \ge N_0(C)$ we have
\begin{multline*}
\mathbb{Q}\left\{\sup_{k 2^{-2m} \in [0, T]} |T^*_{m+1}(2k) 2^{-2(m+1)} - k 2^{-2m}| \ge (\alpha C m 2^{-2m} T \log_* T)^{\frac12} \right\} \\
\le (T 2^{2m})^{1-C} ,
\end{multline*}
where $\alpha = \frac12 + \log 2$ and $\log_* T := \max(1, \log T)$.
\end{lem}
\begin{proof}
We are going to use Lemma \ref{le:largedev} for the i.i.d. terms
\[
Y_j = (T^*_{m+1}(2j) - T^*_{m+1}(2j-2) - 4)/2 \qquad (j\ge 1) .
\]
By (\ref{eq:E_Tmlazy}) and (\ref{eq:mgf_Tmlazy}), these terms are standardized and have a finite moment generating function in a neighborhood of the origin. Take $N = T 2^{2m}$, with $m \ge 1$ fixed. Then
\[
2CN\log N \le 2(1 + \log 4) C m 2^{2m} T \log_* T ,
\]
and
\begin{multline*}
\mathbb{Q}\left\{\sup_{k 2^{-2m} \in [0, T]} |T^*_{m+1}(2k) - 4k|/2 \ge (2 (1 + \log 4) C m 2^{2m} T \log_* T)^{\frac12} \right\} \\
\le (T 2^{2m})^{1-C} ,
\end{multline*}
which is equivalent to the statement of the lemma.
\end{proof}

A consequence of the next lemma is that the sequence $(B_m)_{m \ge 0}$ a.s.\ uniformly converges on any finite interval as $m \to \infty$.
\begin{lem}\label{le:apprconv}
Let the sequence of approximations $(B_m)_{m \ge 0}$ be defined by (\ref{eq:Bm}). For any $C > 1$ there exists a positive $N_1(C)$ such that for any $T > 0$, $m \ge 1$, $T 2^{2m} \ge N_1(C)$ we have
\begin{multline*}
\mathbb{Q}\left\{\sup_{j \ge 1} \sup_{t \in [0, T]} |B_{m+j}(t) - B_m(t)| \ge 39 \cdot C m^{\frac34} 2^{-\frac{m}{2}} T_*^{\frac14} (\log_* T)^{\frac34} \right\} \\
\le \frac{5}{1 - 4^{1-C}} (T 2^{2m})^{1-C} ,
\end{multline*}
where $T_* := \max(1, T)$ and $\log_* T := \max(1, \log T)$.
\end{lem}
\begin{proof}
\textsc{Step 1.}
Fix $m$ and consider the difference of the $m$th and $(m+1)$th approximations at the time instants $t_k := k 2^{-2m} \in [0, T]$:
\begin{multline*}
B_{m+1}(t_k) - B_m(t_k) = B_{m+1}(4k 2^{-2(m+1)}) - B_{m+1}(T^*_{m+1}(2k) 2^{-2(m+1)}) \\
= 2^{-m-1} (\tilde{L}_{m+1}(4k) - \tilde{L}_{m+1}(T^*_{m+1}(2k))) .
\end{multline*}
Let $D_{T,m} := C m^{\frac34} 2^{-\frac{m}{2}} T^{\frac14} (\log_* T)^{\frac34}$. Then
\begin{multline}\label{eq:appr_consec1}
\mathbb{Q}\left\{\sup_{t_k \in [0, T]} |B_{m+1}(t_k) - B_m(t_k)| \ge \beta \, D_{T,m} \right\}
\le \mathbb{Q}(A_{T,m}) \\
+ \sum_{1 \le k \le T 2^{2m}} \mathbb{Q}\left\{ \sup_{\{j : |j-4k| \le N'\}} |\tilde{L}_{m+1}(j) - \tilde{L}_{m+1}(4k)| \ge \beta 2^{m+1} D_{T,m} \right\} ,
\end{multline}
where
\[
A_{T,m} := \left\{\sup_{t_k \in [0, T]} |T^*_{m+1}(2k) - 4k| \ge N' \right\}, \quad N' := 4(\alpha C m 2^{2m} T \log_* T)^{\frac12} ,
\]
$\alpha = \frac12 + \log 2$, and $\beta$ will be chosen below. Now we can apply Lemma \ref{le:timeconv} to the first term on the right hand side of (\ref{eq:appr_consec1}) and Hoeffding's inequality (\ref{eq:Hoeffding}) to the second. We apply Hoeffding's inequality to the partial sums $Z_n = \tilde{L}_m(n)$, so $\mathbb{E}(Z_n) = 0$, $a_j = -b_j = 1$, and for $|j-4k| \le N'$ we have
\begin{multline*}
\mathbb{Q}\left\{|\tilde{L}_{m+1}(j) - \tilde{L}_{m+1}(4k)| \ge x \sqrt{N'} \right\} \\
\le \mathbb{Q}\left\{|\tilde{L}_{m+1}(j) - \tilde{L}_{m+1}(4k)| \ge x \sqrt{|j - 4k|} \right\} \le 2 e^{-\frac{x^2}{2}} .
\end{multline*}

Suppose $T > 0$, $m \ge 1$, $T 2^{2m} \ge N_0(C)$. Then Lemma \ref{le:timeconv} and (\ref{eq:appr_consec1}) imply that
\begin{multline}\label{eq:appr_consec2}
\mathbb{Q}\left\{\sup_{t_k \in [0, T]} |B_{m+1}(t_k) - B_m(t_k)| \ge \beta \, D_{T,m} \right\} \\
\le (T 2^{2m})^{1-C} + T 2^{2m} 2 N' 2 e^{-\frac{x^2}{2}} ,
\end{multline}
where $x$ is chosen so that $x \sqrt{N'} = \beta 2^{m+1} D_{T,m}$. It follows that
\begin{equation}\label{eq:xfirst}
x = \frac{\beta 2^{m+1} C m^{\frac34} 2^{-\frac{m}{2}} T^{\frac14} (\log_* T)^{\frac34}}{2 (\alpha C m 2^{2m} T \log_* T)^{\frac14}} = \frac{\beta}{\alpha^{\frac14}} C^{\frac34} m^{\frac12} (\log_* T)^{\frac12}.
\end{equation}

On the other hand, we demand that $e^{-\frac{x^2}{2}} \le (N')^{-C'}$, and $C' > 1$ is chosen so that $(N')^{1-C'} \le (T 2^{2m})^{(1-C')/2} \le (T 2^{2m})^{-C}$. This implies $C' \ge 2C+1$, so $C' = 3C$ is a suitable choice. There exists a positive $N_1(C) \ge N_0(C)$ such that
\[
(2 C' \log N')^{\frac12} \le (2 C' \log(T 2^{2m}))^{\frac12} \le (6 (1 + \log 4) C m \log_* T)^{\frac12}
\]
if $T 2^{2m} \ge N_1(C)$. Thus
\begin{equation}\label{eq:xsecond}
x \ge (12 \alpha C m \log_* T)^{\frac12}
\end{equation}
satisfies our demand. Comparing (\ref{eq:xfirst}) and (\ref{eq:xsecond}) implies $\beta \ge 2 \sqrt{3} \alpha^{\frac34}$, e.g. $\beta = 4$ is good.

Thus (\ref{eq:appr_consec2}) gives the result of \textsc{Step 1}:
\begin{equation}\label{eq:appr_consec_res}
\mathbb{Q}\left\{\sup_{t_k \in [0, T]} |B_{m+1}(t_k) - B_m(t_k)| \ge 4 \, D_{T,m} \right\}
\le 5 (T 2^{2m})^{1-C} .
\end{equation}

\textsc{Step 2.}
Let $D^*_{T,m} = C m^{\frac34} 2^{-\frac{m}{2}} T_*^{\frac14} (\log_* T)^{\frac34} \ge D_{T,m}$. By (\ref{eq:appr_consec_res}), $|B_{m+1}(t_k) - B_m(t_k)| < 4 \, D^*_{T,m}$ on an event $A_{T,m}$ such that $\mathbb{Q}(A_{T,m}) \ge 1 - 5 (T 2^{2m})^{1-C}$. With $m$ fixed, consider an interval $[t_k, t_{k+1}]$. First, $|B_m(t_{k+1}) - B_m(t_k)| \le 2^{-m} \le 2^{-\frac12} D^*_{T,m}$. Second, $B_{m+1}$ makes $4$ steps of magnitude $\le 2^{-m-1}$ on this interval. On the event $A_{T,m}$, using simple geometry, the maximum deviation between $B_{m+1}$ and $B_m$ at the instant $t_{k+\frac14}$ (and similarly, at $t_{k+\frac34}$) cannot exceed $4 D^*_{T,m} + 2^{-m-1} + \frac14 \cdot 2^{-m} \le (4 + \frac34 \cdot 2^{-\frac12}) D^*_{T,m}$. Similarly, at time $t_{k+\frac12}$ the deviation cannot be larger than $4 D^*_{T,m} + 2 \cdot 2^{-m-1} \le (4 + 2^{-\frac12}) D^*_{T,m}$. Hence
\[
\mathbb{Q}\left\{\sup_{t \in [0, T]} |B_{m+1}(t) - B_m(t)| \ge \left(4 + 2^{-\frac12}\right) \, D^*_{T,m} \right\} \le 5 (T 2^{2m})^{1-C} .
\]

Let us use the fact (obtained by Wolfram Mathematica) that for any $m \ge 1$,
\[
\sum_{j=0}^{\infty} (m+j)^{\frac34} 2^{-\frac{m+j}{2}} \le m^{\frac34} 2^{-\frac{m}{2}} \sum_{j=0}^{\infty} (1+j)^{\frac34} 2^{-\frac{j}{2}} < \frac{65}{8} \, m^{\frac34} 2^{-\frac{m}{2}} .
\]
Then
\begin{multline*}
\mathbb{Q}\left\{\sup_{j \ge 1} \sup_{t \in [0, T]} |B_{m+j}(t) - B_m(t)| \ge \frac{65}{8} \left(4 + 2^{-\frac12}\right) \, D^*_{T,m} \right\} \\
\le \sum_{j=0}^{\infty} \mathbb{Q}\left\{\sup_{t \in [0, T]} |B_{m+j+1}(t) - B_{m+j}(t)| \ge \left(4 +  2^{-\frac12}\right) \, D^*_{T,m+j} \right\} \\
\le \sum_{j=0}^{\infty} 5 (T 2^{2(m+j)})^{1-C} = \frac{5}{1 - 4^{1-C}} (T 2^{2m})^{1-C} .
\end{multline*}
This proves the lemma.
\end{proof}

We only sketch the proof of the next theorem, because it is pretty standard; moreover, it is essentially the same as the proof of \cite[Theorem 3]{SZA96}.
\begin{thm}\label{th:Brown}
As $m \to \infty$, the approximations $(B_m(t))_{t \ge 0}$ a.s.\ converge to Brownian motion $(B(t))_{t \ge 0}$. (The variance of an increment $B(t) - B(s)$ is $(t - s)/2$.) Moreover,
\begin{multline}\label{eq:Brown}
\mathbb{Q}\left\{\sup_{t \in [0, T]} |B(t) - B_m(t)| \ge 39 \cdot C m^{\frac34} 2^{-\frac{m}{2}} T_*^{\frac14} (\log_* T)^{\frac34} \right\} \\
\le \frac{5}{1 - 4^{1-C}} (T 2^{2m})^{1-C} ,
\end{multline}
\end{thm}
\begin{proof}
(Sketch.)

The a.s.\ convergence follows from Lemma \ref{le:apprconv} by the Borel--Cantelli lemma. (\ref{eq:Brown}) also follows from Lemma \ref{le:apprconv}.

For any $0 \le s < t$, as $m \to \infty$,
\begin{multline*}
B(t) - B(s) \sim B_m(t) - B_m(s) \sim 2^{-m} \left(\tilde{L}_m(\lfloor t 2^{2m} \rfloor) - \tilde{L}_m(\lfloor s 2^{2m} \rfloor)\right) \\
= \left(\frac{t-s}{2}\right)^{\frac12} \frac{\tilde{L}_m(\lfloor t 2^{2m} \rfloor) - \tilde{L}_m(\lfloor s 2^{2m} \rfloor)}{\left(\frac{(t-s)2^{2m}}{2}\right)^{\frac12}} \sim \mathcal{N}\left(0, \frac{t-s}{2}\right),
\end{multline*}
where $\mathcal{N}(\mu, \sigma^2)$ denotes a Gaussian random variable with expectation $\mu$ and variance $\sigma^2$.

Finally, for any $0 \le s < t \le u < v$, the increments $B(v) - B(u)$ and $B(t) - B(s)$ are independent, because
for any $m \ge 1$, the approximating increments
\[
2^{-m} \left(\tilde{L}_m(\lfloor v 2^{2m} \rfloor) - \tilde{L}_m(\lfloor u 2^{2m} \rfloor)\right) \text{ and } 2^{-m} \left(\tilde{L}_m(\lfloor t 2^{2m} \rfloor) - \tilde{L}_m(\lfloor s 2^{2m} \rfloor)\right)
\]
are independent.
\end{proof}

%%%%%%%%%%%%%%%%%%%%%%%%%%%%%%%%%%%%%%%%%%%%%%%%%%%%%%%%%%%%%%%%%%%

\section{A strong approximation of Brownian motion based on coupled random walks} \label{se:coupling}

Take the measurable space $(\Omega, \mathcal{F})$ as in Subsection \ref{sse:Brownian}. Fix a positive integer $T$. The value of the positive integer $m$ should be large enough, as specified later. We also need an infinite matrix of independent, simple, symmetric random walk steps $(Z_j(k))_{j \ge 0, k \ge 1}$, similarly defined as the matrix of lazy random walk steps in Subsection \ref{sse:lazy}:
\[
\mathbb{Q}^*\left(Z_j(k) = \frac12\right) = \mathbb{Q}^*\left(Z_j(k) = -\frac12\right) = \frac12 .
\]

First, take complex measure walk steps $X_r$ $(1 \le r \le 2T)$ on the set $\{1, 2, \dots, 2T\}$, see Section \ref{se:Complex}. Augment this with simple, symmetric random walk steps $(Z_0(k))_{k \ge 0}$ on the set $\{2T+1, 2T+2, 2T+3, \dots \}$ and denote the so obtained infinite walk by $(S_0(n))_{n  \ge 0}$: $S_0(0) = 0$,
\[
S_0(n) = \sum_{r = 1}^n X_r \quad (1 \le n \le 2T), \quad S_0(2T + k) = S_0(2T) + \sum_{r = 1}^k Z_0(r) \quad (k \ge 1) .
\]

Second, take complex measure walk steps $X_r$ $(2T+1 \le r \le 2T + 2^4 T)$ on $\{1,2, 3, \dots, 2^4 T\}$. Augment this with simple, symmetric random walk steps $(Z_1(k))_{k \ge 0}$ on the set $\{2^4 T+1, 2^4 T+2, 2^4 T+3, \dots \}$, and denote the so obtained infinite walk by $(S_1(n))_{n  \ge 0}$: $S_1(0) = 0$,
\begin{multline*}
S_1(n) = \sum_{r = 2T+1}^{2T+n} X_r \quad (1 \le n \le 2^4 T), \\
S_1(2^4 T + k) = S_1(2^4 T) + \sum_{r = 1}^k Z_1(r) \quad (k \ge 1) .
\end{multline*}

And so on, at the $m$th walk, take complex measure walk steps $X_r$ $(1+2T+2^4T+2^7T+\dots+2^{3m-2}T \le r \le 2T+2^4T+2^7T+\dots+2^{3m+1}T)$ on the set $\{1, 2, \dots, 2^{3m+1}T\}$. Augment this with simple, symmetric random walk steps $(Z_m(k))_{k \ge 0}$ on the set $\{2^{3m+1}T+1, 2^{3m+1}T+2, 2^{3m+1}T+3, \dots \}$, and denote the so obtained infinite walk by $(S_m(n))_{n  \ge 0}$: $S_m(0) = 0$,
\begin{multline*}
S_m(n) = \sum_{r = 2^{3m-2}T+1}^{2^{3m-2}T+n} X_r \quad (1 \le n \le 2^{3m+1}T), \\
S_m(2^{3m+1}T + k) = S_m(2^{3m+1}T) + \sum_{r = 1}^k Z_m(r) \quad (k \ge 1) .
\end{multline*}

After the $m$th walk, each walk will be simple, symmetric random walk. If $j > m$, $S_j(0)=0$,
\[
S_j(k) = \sum_{r = 1}^k Z_j(r) \quad (k \ge 1) .
\]

In sum, we needed $N = \frac27 (8^{m+1} - 1) T$ complex measure walk steps $X_r$. For the coupled random walks we are going to use only every second steps. As we saw in Subsection \ref{sse:Markov}, there exists a corresponding coupled probability distribution $\mathbb{P}_{\{0 \dots 2N\}}$ on $(\mathbb{R}^{\{0 \dots 2N\}}, \mathcal{B}^{\{0 \dots 2N\}})$, on the even non-negative integers, for the finite triangular matrix $(S_j(2k))$ $(0 \le j \le m, 0 \le k \le 2^{3j+1}T)$ of random walks defined above. Each row of this matrix was augmented by infinitely many simple symmetric random walk steps $Z_m(r)$. This way, we are given the probability measure $\mathbb{P}_{\{0 \dots 2N\}} \times \mathbb{Q}^*$ on the set of even positive integers, that is, on $(\mathbb{R}^{2\mathbb{N}}, \mathcal{B}^{2\mathbb{N}})$. Since the continuous sample functions will be functions of the elements of infinite matrix of the steps defined above, the probability measure can be extended to $(\Omega, \mathcal{B})$ and will be denoted by $\mathbb{P}$.

Because of the symmetry (\ref{eq:symprob}) of $S_j(2k)$ about $0$, $\mathbb{E}S_j(2k) = 0$ for any $j \ge 0$ and $k \ge 0$. By (\ref{eq:binom_even}) and (\ref{eq:tailterms}), there exists a constant $C_7 > 0$ such that
\[
(1 - C_7 k^{-\frac13}) \frac{k}{2} \le \text{Var}(S_j(2k)) \le (1 + C_7 k^{-\frac13}) \frac{k}{2}
\]
for any $j \ge 0$ and $k \ge 1$. By assumption (i) in Subsection \ref{sse:Markov}, $S_j$ has bounded increments, $|S_j(2k+2) - S_j(2k)| \le 2$.

Now we are going to define stopping times, similarly to (\ref{eq:stopt_star}). Take $j \ge 1$. Define $T_j(0) := 0$, and for any $k \ge 0$,
\begin{equation}\label{eq:stopt}
T_j(k+1) := \inf\{2n : 2n > T_j(k), |S_j(2n) - S_j(T_j(k))| \ge 1\} .
\end{equation}
To determine the distribution and moments of these stopping times, at least when they are large enough, consider an integer $n \ge n_0$, where $n_0$ is defined in Theorem \ref{th:Markov}. Put
\[
\tau_n := \inf \{\ell \ge 1 : |S_1(2n + 2 \ell) - S_1(2n)| \ge 1 \} .
\]
Below it is supposed that $n$ is sufficiently large. Then
\begin{multline*}
\mathbb{P}(\tau_n = 1) = \mathbb{P}\left(|S_1(2n)| \le \frac13 n^{\frac23}\right) \mathbb{P}(|S_1(2n + 2) - S_1(2n)| \ge 1) \\
+ \mathbb{P}\left(|S_1(2n)| > \frac13 n^{\frac23}\right) \mathbb{P}(|S_1(2n + 2) - S_1(2n)| \ge 1) \end{multline*}
So by Theorems \ref{th:norm_factor}, \ref{th:Markov} and assumption (ii) in Subsection \ref{sse:Markov}, we obtain that
\[
\mathbb{P}(\tau_n = 1) \le \frac12 + 2 C_6 n^{-\frac13} + C_5 n^{-\frac13} + (1 + C_4 (2n)^{-\frac13}) 4 \sqrt{\frac{n}{\pi}} e^{-\frac29 (2n)^{\frac13}} \le \frac12 + C_8 n^{-\frac13} ,
\]
where $C_8 > 0$ is a suitable constant. Similarly,
\begin{multline*}
\mathbb{P}(\tau_n = 1) \ge \left(\frac12 - 2 C_6 n^{-\frac13}\right) \left(1 - (1 + C_4 (2n)^{-\frac13}) 4 \sqrt{\frac{n}{\pi}} e^{-\frac29 (2n)^{\frac13}} \right) \\
\ge \frac12 - C_8 n^{-\frac13} .
\end{multline*}
These imply that for any $k \ge 1$,
\[
\left(\frac12 - C_8 n^{-\frac13}\right)^k \le \mathbb{P}(\tau_n = k) \le \left(\frac12 + C_8 n^{-\frac13}\right)^k .
\]
(Compare with Subsection \ref{sse:lazy}.) So for the first two moments we get
\begin{gather*}
2 \frac{1 - 2 C_8 n^{-\frac13}}{\left(1 + 2 C_8 n^{-\frac13}\right)^2} \le \mathbb{E}\tau_n \le 2 \frac{1 + 2 C_8 n^{-\frac13}}{\left(1 - 2 C_8 n^{-\frac13}\right)^2} , \\
6 \frac{(1 - 2 C_8 n^{-\frac13})(1 - \frac23 C_8 n^{-\frac13})}{\left(1 + 2 C_8 n^{-\frac13}\right)^3} \le \mathbb{E}\tau_n^2 \le 6 \frac{(1 + 2 C_8 n^{-\frac13})(1 + \frac23 C_8 n^{-\frac13})}{\left(1 - 2 C_8 n^{-\frac13}\right)^3} .
\end{gather*}
That is, there exists a constant $C_9 > 0$ such that
\begin{gather}
2 (1 - C_9 n^{-\frac13}) \le \mathbb{E}\tau_n \le 2 (1 + C_9 n^{-\frac13}) , \label{eq:Etau} \\
2 (1 - C_9 n^{-\frac13}) \le \text{Var}(\tau_n) \le 2 (1 + C_9 n^{-\frac13}) \label{eq:Vartau} .
\end{gather}
Also,
\begin{gather*}
\frac{e^u(1 - 2 C_8 n^{-\frac13})}{2 - e^u (1 - 2 C_8 n^{-\frac13})} \le \mathbb{E}e^{u \tau_n} \le \frac{e^u (1 + 2 C_8 n^{-\frac13})}{2 - e^u (1 + 2 C_8 n^{-\frac13})} , \\
\frac{e^{-\frac{u}{\sqrt{2}}(1 + C_{10} n^{-\frac13})} (1 - C_{10} n^{-\frac13})}{2 - e^{\frac{u}{\sqrt{2}} (1 - C_{10} n^{-\frac13})} (1 - C_{10} n^{-\frac13})} \le \mathbb{E}e^{u \tau_n^*} \le \frac{e^{-\frac{u}{\sqrt{2}}(1 - C_{10} n^{-\frac13})} (1 + C_{10} n^{-\frac13})}{2 - e^{\frac{u}{\sqrt{2}} (1 + C_{10} n^{-\frac13})} (1 + C_{10} n^{-\frac13})} ,
\end{gather*}
where $\tau_n^* := (\tau_n - \mathbb{E}\tau_n) (\text{Var}(\tau_n))^{-\frac12}$ and $C_{10} > 0$ is a suitable constant. It follows that the moment generating function $\mathbb{E}e^{u \tau_n^*}$ is finite in a neighborhood $|u| \le u_0$, $u_0 > 0$, if $n \ge n_1 = n_1(C_{10}) \ge n_0$.

At this point it is natural trying to generalize Lemma \ref{le:largedev} to the current situation. The random variables we have to consider are $(\tau^*_j(k))_{k \ge k_1}$ $(j \ge j_1)$, where
\[
\tau_j(k) := T_j(k) - T_j(k-1), \qquad \tau^*_j(k) := (\tau_j(k) - \mathbb{E}\tau_j(k)) (\text{Var}(\tau_j(k)))^{-\frac12},
\]
and the positive integers $k_1$ and $j_1$ are chosen so that
\begin{equation}\label{eq:k1_def}
T_j(k-1) \ge n_1 \quad \mathbb{P}\text{ - a.s. } \text{ for any } k \ge k_1 \text{ and } j \ge j_1 .
\end{equation}
We suppose that $m \ge j_1$. By the previous computations, the random variables $\tau^*_j(k)$ are asymptotically independent and identically distributed w.r.t.\ $\mathbb{P}$, and otherwise satisfy all the other conditions of Lemma \ref{le:largedev}. We suppose now that the statement of the lemma is valid to this case as well.

Now we are ready to begin a \emph{"twist and shrink"} construction of Brownian motion, based on the the random walks $(S_j(2k))_{k \ge 0}$ $(j \ge 0)$. It is going to be a slight modification of the one discussed in Subsection \ref{sse:Brownian}.

\textsc{(1) Twisting.} The first approximation is based on $(S_0(n))_{n \ge 0}$. It is not twisted: $\tilde{S}_0(n) = S_0(n)$ for all $n \ge 0$. By induction, suppose that the $j$th twisted random walk $(\tilde{S}_j(n))_{n \ge 0}$ is already defined where $j \ge 0$. The next twisted walk will be based on $S_{j+1}$ and $\tilde{S}_j$. Note that reflections will be based only on even time instants, although a bridge reflected will contain the included odd time instants as well.

For $k \ge 1$ define $\tilde{Y}_j(k) := \tilde{S}_j(2k) - \tilde{S}_j(2k-2)$.

(a) If $\tilde{Y}_j(k) = 0$ and
\begin{multline*}
\epsilon_k := \{S_{j+1}(T_{j+1}(2k-1))  - S_{j+1}(T_{j+1}(2k-2))\} \\
\times \{S_{j+1}(T_{j+1}(2k)) - S_{j+1}(T_{j+1}(2k-1))\} = \pm 1 ,
\end{multline*}
then we set
\[
\tilde{X}_{j+1}(n) := \left\{
\begin{gathered}
S_{j+1}(n) - S_{j+1}(n-1) \text{ if } T_{j+1}(2k-2) < n \le T_{j+1}(2k-1) , \\
- \epsilon_k (S_{j+1}(n) - S_{j+1}(n-1)) \text{ if } T_{j+1}(2k-1) < n \le T_{j+1}(2k) .
\end{gathered} \right.
\]

(b) If $\tilde{Y}_j(k) = \pm 1$, let
\begin{align*}
\epsilon_{k,1} &:= \tilde{Y}_j(k) \{S_{j+1}(T_{j+1}(2k-1))  - S_{j+1}(T_{j+1}(2k-2))\} = \pm 1 , \\
\tilde{X}_{j+1}(n) &:= \epsilon_{k,1} \tilde{Y}_{j+1}(n) \text{ if } T_{j+1}(2k-2) < n \le T_{j+1}(2k-1) ; \\
\epsilon_{k,2} &:= \tilde{Y}_j(k) \{S_{j+1}(T_{j+1}(2k)) - S_{j+1}(T_{j+1}(2k-1))\} = \pm 1 , \\
\tilde{X}_{j+1}(n) &:= \epsilon_{k,2} \tilde{Y}_{j+1}(n) \text{ if } T_{j+1}(2k-1) < n \le T_{j+1}(2k) .
\end{align*}

Finally,
\[
\tilde{S}_{j+1}(0) := 0, \qquad \tilde{S}_{j+1}(n) := \sum_{r=1}^n \tilde{X}_{j+1}(r) \qquad (n \ge 1) .
\]

It is important that a complex measure walk and also, a simple, symmetric random walk preserve their basic properties, because as stochastic processes they are not affected by the applied reflections.

\textsc{(2) Shrinking.} The sample functions of the $j$th approximation make steps of length $2^{-j}$ in time-steps $2^{-2j}$. Thus
\begin{equation}\label{eq:Bj_co}
B_j(n 2^{-2j}) := 2^{-j} \tilde{S}_j(n) \qquad (n \ge 0)
\end{equation}
and for any real $t \in (n 2^{-2j}, (n+1) 2^{-2j})$,
\[
B_j(t) := B_j(n 2^{-2j}) + 2^{2j} \{B_j((n+1) 2^{-2j}) - B_j(n 2^{-2j})\} (t - n 2^{-2j}).
\]
The \emph{refinement property} of the approximation is
\begin{equation}\label{eq:refine_co}
B_{j+1}(T_{j+1}(2k) 2^{-2(j+1)}) = B_j(2k 2^{-2j}) \qquad (k, j \ge 0).
\end{equation}
Its meaning is that $B_{j+1}$ takes the same values as $B_j$, in the same order, at stopping times $T_{j+1}(2k) 2^{-2(j+1)}$ that correspond to the times $2k 2^{-2j}$. There is a time lag in general, but these lags a.s.\ uniformly converge to $0$ on any finite time interval as $j \to \infty$, as easily follows from the next lemma by the Borel--Cantelli lemma. The next lemma is a slight modification of Lemma \ref{le:timeconv}, thus its proof is omitted.
\begin{lem}\label{le:timeconv_co}
Let the stopping times $T_j(k)$ be defined by (\ref{eq:stopt}). For any $C > 1$ there exists a positive $N_0(C)$ such that for any $T > 0$, $j \ge j_1$, $T 2^{2j} \ge N_0(C)$ we have
\begin{multline*}
\mathbb{P}\left\{\sup_{2k 2^{-2j} \in [t_1, T]} |T_{j+1}(2k) 2^{-2(j+1)} - 2k 2^{-2j}| \ge (\alpha_0 C j 2^{-2j} T \log_* T)^{\frac12} \right\} \\
\le (T 2^{2j})^{1-C} ,
\end{multline*}
where $t_1 := k_1 2^{-2j}$, $\alpha_0 > 0$ is a suitable constant, and $\log_* T := \max(1, \log T)$.
\end{lem}

By (\ref{eq:binom_even}) and the first part of the present section, the increments $S_j(2k) - S_j(2k-2)$ are asymptotically independent and identically distributed random variables w.r.t.\ $\mathbb{P}$, with zero expectation, and are bounded by $2$, at least when $k$ is large enough, $k \ge k_1$. (Suppose that $k_1$ was chossen so.) Assuming that Hoeffding's inequality (\ref{eq:Hoeffding}) can be generalized to this case as well, we can state the following slight modification of Lemma \ref{le:apprconv} whose proof is therefore omitted.
\begin{lem}\label{le:Bjconv}
Let the sequence of approximations $(B_j)_{j \ge 0}$ be defined by (\ref{eq:Bj_co}). For any $C > 1$ there exists a positive $N_1(C)$ such that for any $T > 0$, $j \ge j_1$, $T 2^{2j} \ge N_1(C)$ we have
\begin{multline*}
\mathbb{P}\left\{\sup_{k \ge 1} \sup_{t \in [t_1, T]} |B_{j+k}(t) - B_j(t)| \ge \beta_0 C j^{\frac34} 2^{-\frac{j}{2}} T_*^{\frac14} (\log_* T)^{\frac34} \right\} \\
\le \frac{5}{1 - 4^{1-C}} (T 2^{2j})^{1-C} ,
\end{multline*}
where $t_1 = k_1 2^{-2j}$, $\beta_0 > 0$ is a suitable constant, $T_* := \max(1, T)$ and $\log_* T := \max(1, \log T)$.
\end{lem}

As in the case of Theorem \ref{th:Brown}, this lemma implies the following theorem, whose proof is again omitted.
\begin{thm}\label{th:Brown_co}
As $j \to \infty$, the approximations $(B_j(t))_{t \ge 0}$ $\mathbb{P}$-a.s.\ converge to Brownian motion $(B(t))_{t \ge 0}$. (The variance of an increment $B(t) - B(s)$ is $(t - s)/4$.) Moreover,
\begin{multline}\label{eq:Brown_co}
\mathbb{P}\left\{\sup_{t \in [k_1 2^{-2m}, T]} |B(t) - B_m(t)| \ge \beta_0 \cdot C m^{\frac34} 2^{-\frac{m}{2}} T_*^{\frac14} (\log_* T)^{\frac34} \right\} \\
\le \frac{5}{1 - 4^{1-C}} (T 2^{2m})^{1-C} ,
\end{multline}
where $k_1$ is defined by (\ref{eq:k1_def} and the other constants are the same as in Lemma \ref{le:Bjconv}.
\end{thm}
The really important consequence of this theorem is that the random walk $(B_m(t))_{t_1 \le t \le T}$ coupled to a complex measure walk is arbitrarily close to a Brownian motion a.s.\ if $m$ is large enough.

%%%%%%%%%%%%%%%%%%%%%%%%%%%%%%%%%%%%%%%%%%%%%%%%%%%%%%%%%%%%%%%%%%%

\section{Convergence and the Schr\"odinger equation} \label{se:conv_Schr}

Temporarily fix an integer $m \ge j_1$, where $j_1$ is defined by (\ref{eq:k1_def}). Introduce the notations $\Delta x := 2^{-m}$, $\Delta t := 2^{-2m}$, and $t_r := r 2^{-2m}$, where $r \ge 0$ is an integer. Now we extend the discrete path integral (\ref{eq:psik_compl}) to the twisted and shrunken random walks $B^x_m(t_r) := x + B_m(t_r)$ $(x \in \mathbb{R})$. It means that if $x_j := j 2^{-m}$ $(j \in \mathbb{Z})$ and $B_m(t_r) = x_j$, then $B_m(t_{r+1})$ is concentrated on $\{x_{j-1}, x_j, x_{j+1}\}$ and
\[
\mu(B_m(t_{r+1}) = x_{j+1}) = \mu(B_m(t_{r+1}) = x_{j-1}) = \frac{i}{2}, \quad \mu(B_m(t_{r+1}) = x_{j}) = 1 - i .
\]
Take functions $V,g : \mathbb{R} \to \mathbb{R}$ and define
\begin{equation}\label{eq:psim}
\psi_m(t_k, x) := \mathbb{E}_{\mu} \left\{\exp\left(-i \sum_{r=0}^{k-1} V(B^x_m(t_r)) \Delta t \right) g(B^x_m(t_k))\right\} .
\end{equation}
By the argument of Lemma \ref{le:compl}, this discrete path integral is the unique solution of the discrete Schr\"odinger equation
\begin{multline}\label{eq:disc_Schr}
\frac{1}{i} \frac{\psi_m(t_{k+1},x) - \psi_m(t_k,x)}{\Delta t} \\
 = \frac12 e^{-iV(x)\Delta t} \, \frac{\psi_m(t_k, x + \Delta x) - 2 \psi_m(t_k, x) + \psi_m(t_k, x - \Delta x)}{(\Delta x)^2} \\
 + \frac{1}{i} \frac{e^{-iV(x) \Delta t} - 1}{\Delta t} \psi_m(t_k, x) , \quad \psi_m(0, x) = g(x) .
\end{multline}

\begin{thm}\label{th:Schrod}
Suppose that $V,g \in C^2$ and $V, V', V'', g, g', g''$ are bounded. Then there exists a function $\psi \in C^{1,2}(\mathbb{R}^+ \times \mathbb{R} \to \mathbb{C})$ that solves Schr\"odinger equation (\ref{eq:Sch}) for $(t, x) \in [0, T] \times \mathbb{R}$ and for any $\epsilon > 0$ there exists an $m_0$ such that
\[
\sup_{(t, x) \in [0, T] \times \mathbb{R}} |\psi_m(t^{(m)}, x) - \psi(t, x)| < \epsilon \quad \text{ if } \quad m \ge m_0 ,
\]
where $t^{(m)}  := \lfloor t 2^{2m} \rfloor 2^{-2m}$.
\end{thm}
\begin{proof}
Define the complex valued random variables
\begin{align*}
Y_m(t_k, x) &:= \exp\left(-i \sum_{r=0}^{k-1} V(B^x_m(t_r)) \Delta t \right) g(B^x_m(t_k)), \\
Y(t, x) &:= \exp\left(-i \int_0^t V(B^x(s)) ds \right) g(B^x(t)) \quad ((t, x) \in \mathbb{R}^+ \times \mathbb{R}) ,
\end{align*}
where $B^x(t) := x + B(t)$ and $B$ is defined in Theorem \ref{th:Brown_co}. Then by our assumptions it follows that for any $\epsilon > 0$ there exists an $m_0$ such that
\begin{align*}
\sup_{(t, x) \in [0, T] \times \mathbb{R}} |Y_m(t^{(m)}, x) - Y(t, x)| &< \epsilon \quad \mathbb{P}\text{-a.s.}\quad \text{ if } \quad m \ge m_0 ,\\
\sup_{(t, x) \in \mathbb{R}^+ \times \mathbb{R}} |Y_{m_1}(t^{(m_1)}, x) - Y_{m_2}(t^{(m_2)}, x)| &< \epsilon \quad \mathbb{P}\text{-a.s.}\quad \text{ if } \quad m_1, m_2 \ge m_0 .
\end{align*}
Hence
\begin{equation}\label{eq:psi_diff}
\sup_{(t, x) \in [0, T] \times \mathbb{R}} |\psi_{m_1}(t^{(m_1)}, x) - \psi_{m_2}(t^{(m_2)}, x)| < \epsilon \quad \text{ if } \quad m_1, m_2 \ge m_0 .
\end{equation}
(The value of $m_0$ could be different from line to line.)

Since
\begin{multline*}
(\partial_{xx} \psi_m)(t_k, x) \\
= \mathbb{E}_{\mu} \left\{\exp\left(-i \sum_{r=0}^{k-1} V(B^x_m(t_r)) \Delta t \right) \left[\left(-i \sum_{r=0}^{k-1} V'(B^x_m(t_r)) \Delta t\right) \right.\right. \\
\left.\left. \times \left(-i \sum_{r=0}^{k-1} V'(B^x_m(t_r)) \Delta t \; g(B^x_m(t_k))\right) -i \sum_{r=0}^{k-1} V''(B^x_m(t_r)) \Delta t \; g(B^x_m(t_k)) \right.\right. \\
\left.\left. -i \sum_{r=0}^{k-1} V'(B^x_m(t_r)) \Delta t \; g'(B^x_m(t_k)) + g''(B^x_m(t_k)) \right]\right\} ,
\end{multline*}
by our assumptions we similarly obtain that
\begin{equation}\label{eq:psixx_diff}
\sup_{(t, x) \in [0, T] \times \mathbb{R}} |(\partial_{xx}\psi_{m_1})(t^{(m_1)}, x) - (\partial_{xx}\psi_{m_2})(t^{(m_2)}, x)| < \epsilon \quad \text{ if } \quad m_1, m_2 \ge m_0 .
\end{equation}

(\ref{eq:psi_diff}) and (\ref{eq:psixx_diff}) imply that there exists a function $\psi \in C^{1,2}(\mathbb{R}^+ \times \mathbb{R} \to \mathbb{C})$ such that
\begin{align*}
\sup_{(t, x) \in [0, T] \times \mathbb{R}} |\psi_{m}(t^{(m)}, x) - \psi(t, x)| &< \epsilon , \\
\sup_{(t, x) \in [0, T] \times \mathbb{R}} |(\partial_{xx}\psi_{m})(t^{(m)}, x) - (\partial_{xx}\psi)(t, x)| &< \epsilon, \quad \text{ if } \quad m \ge m_0 .
\end{align*}
Thus
\begin{multline*}
\sup_{(t, x) \in [0, T] \times \mathbb{R}} \left|\frac12 e^{-iV(x)\Delta t} \, \frac{\psi_m(t_k, x + \Delta x) - 2 \psi_m(t_k, x) + \psi_m(t_k, x - \Delta x)}{(\Delta x)^2} \right. \\
\left.  + \frac{1}{i} \frac{e^{-iV(x) \Delta t} - 1}{\Delta t} \psi_m(t_k, x) - \frac12 (\partial_{xx}\psi)(t, x) + V(x) \psi(t, x)\right| < \epsilon
\end{multline*}
if $m \ge m_0$. By (\ref{eq:disc_Schr}) this implies that $\psi$ is differentiable w.r.t. $t$ and $\psi$ solves Schr\"odinger equation (\ref{eq:Sch}) for $(t, x) \in [0, T] \times \mathbb{R}$.

\end{proof}

Finally, let us see the convergence of a normalized complex measure describing the motion of the twisted and shrunken random walk $B_m$. By Corollary \ref{co:asympt_sol},
\[
\mu(S_{\ell} = j) \sim \sqrt{\frac{2}{\ell}} c_2 e^{i \pi j - (2+i) \frac{j^2}{2\ell}} (-3-4i)^{\frac{\ell}{2}}, \]
where $\sim$ denotes asymptotic equality as $\ell \to \infty$. It means that
\[
\mu(B_m(t_k) = x_j) \sim \sqrt{\frac{2}{t_k}} c_2 \exp\left(i \pi x_j 2^m - (2+i) \frac{x_j^2}{2 t_k}\right) (-3-4i)^{t_k 2^{2m-1}} 2^{-m} .
\]
Using the independence of increments of $B_m$ with the complex measure, it follows that for $0 = t_{k0} < t_{k1} < t_{k2} < \cdots < t_{kd}$,
\begin{multline*}
\mu(B_m(t_{k1}) = x_{j1}, B_m(t_{t2}) = x_{j2}, \dots , B_m(t_{kd}) = x_{jd})  \\
\sim \frac{2^{\frac{d}{2}} c_2^d}{\prod_{r=1}^d (t_{kr} - t_{k(r-1)})} \exp\left(i \pi x_d 2^m - (2+i) \sum_{r=1}^d \frac{(x_{jr} - x_{j(r-1)})^2}{2 (t_{kr} - t_{k(r-1)})}\right) \\
\times (-3-4i)^{t_{kd} 2^{2m-1}} 2^{-md} .
\end{multline*}
Here we introduce a complex normalizing factor, which is somewhat similar to the square root of $Z_{\ell}$ in (\ref{eq:norm_factor}) and which will divide the measure $\mu$:
\[
Z^*_{t_d} := \sqrt{\frac{4 \pi}{2 + i}} c_2 5^{t_d 2^{2m-1}} .
\]
Moreover, we consider only those time instants $t_d$ for which
\[
\left(1 + \frac{1}{\pi} \arctan \frac43\right) t_d 2^{2m-1}
\]
is an even positive integer, and only those points $x_d$ for which $x_d 2^{2m}$ is an even integer. These points are dense on the $t$ and $x$ axis, respectively, as $m \to \infty$. This way we arrive at the time homogeneous and spatially homogeneous \emph{complex transition function}
\[
\nu_t(x, dy) := \sqrt{\frac{2+i}{2 \pi t}} e^{-(2+i) \frac{(y-x)^2}{2t}} dy .
\]
The following argument is taken from \cite{KOL05}. This $\nu$ satisfies the Chapman-Kolmogorov equation and
\[
\int_{\mathbb{R}} \nu_t(x, dy) = 1, \quad \|\nu_t(x, \cdot)\| = \frac{5^{\frac14}}{2^{\frac12}} .
\]
Thus $\nu$ generates a strongly continuous, regular semigroup over $C_0(\mathbb{R})$ (continuous functions vanishing at infinity) by the formula
\[
(T_t f)(x) := \int_{\mathbb{R}} f(y) \nu_t(x, dy) .
\]
Then we introduce the one-point compactification $\dot{\mathbb{R}}$ and for $0 \le s < t$ the path space $\dot{\mathbb{R}}^{[s,t]}$ (infinite product of $[s, t]$ copies). Then $\nu$ defines a family of linear functionals on cylindric functions $Cyl_{[s,t]}$ on the path space by the formula
\begin{multline*}
\nu^x_{s,t}(\phi^f_{t_0, \dots, t_{k+1}}) \\
:= \int_{\mathbb{R}^k} f(x, y_1, \dots ,y_{k+1}) \nu_{t_1 - t_0}(x, dy_1) \nu_{t_2 - t_1}(y_1, dy_2) \dots \nu_{t_{k+1} - t_k}(y_k, dy_{k+1}) ,
\end{multline*}
where $k \ge 0$, $s = t_0 < t_1 < \cdots < t_k < t_{k+1} = t$, $x \in \mathbb{R}$, $f : \dot{\mathbb{R}}^{k+2} \to \mathbb{C}$ is a bounded Borel function, and
\[
\phi^f_{t_0, \dots, t_{k+1}}(y(\cdot)) := f(y(t_0), \dots , y(t_{k+1})) .
\]
Then by continuity, $\nu^x_{s,t}$ can be extended to a unique bounded linear functional on $C(\dot{\mathbb{R}}^{[s,t]})$, and consequently there exists a Borel measure $\mu^x_{s,t}$ on the path space $\dot{\mathbb{R}}^{[s,t]}$ such that
\begin{align*}
\nu^x_{s,t}(F) &= \int F(y(\cdot)) \mu^x_{s,t}(dy(\cdot)) \quad (F \in C(\dot{\mathbb{R}}^{[s,t]})) , \\
(T_t f)(x) &= \int f(y(t))) \mu^x_{s,t}(dy(\cdot)) .
\end{align*}

It is an open question whether the function $\psi(t, x)$ of Theorem \ref{th:Schrod} can be obtained by this Borel measure $\mu^x_{s,t}$ from the complex random variables $Y(t,x)$ or not.

\subsection*{Acknowledgement}
I thank my former graduate student Zolt\'an Feller for his valuable work with computational experiments and for his useful questions and remarks.

%%%%%%%%%%%%%%%%%%%%%%%%%%%%%%%%%%%%%%%%%%%%%%%%%%%%%%%%%%%%%%%%%%%


\begin{thebibliography}{99}

\bibitem{ALB97} Albeverio, S. (1997) Wiener and Feynman path integral and their applications. \emph{Proc. Symp. Appl. Math.} \textbf{52}, 163--194.

\bibitem{CAM60} Cameron, R.H. (1960) A family of integrals serving to connect the Wiener and Feynman integrals. \emph{J. Math. and Phys.} \textbf{39}, 126--140.

\bibitem{CS83} Cameron, R.H. and Storvick, D.A. (1983) A simple definition of the Feyman integrals with applications. \emph{Mem. A.M.S.}, Providence, R. I.

\bibitem{CV80} Chung, K.L. and Varadhan, S. R. S. (1980) Kac functional and Schr\"odinger equation. \emph{Studia Math.} \textbf{68},  249–-260.

\bibitem{CSA93} Cs\'aki, E. (1993) A discrete Feynman-Kac formula. \emph{J. Statist. Plann. Inference} \textbf{34},  63–-73.

\bibitem{EMOT53} Erd\'elyi, A., Magnus, W., Oberhettinger, F. and Tricomi, F.G. (1953) \emph{Higher transcendental functions.}, McGraw-Hill, New York.

\bibitem{FEL68} Feller, W. (1968) \emph{An introduction to probability theory and its applications. Volume I}, Third edition, Wiley, New York.

\bibitem{FEL66} Feller, W. (1966) \emph{An introduction to probability theory and its applications. Volume II}, Wiley, New York.

\bibitem{FEY48} Feynman, R. P. (1948) Space-time approach to non-relativistic quantum mechanics. \emph{Rev. Mod. Phys.} \textbf{20}, 367–-387.

\bibitem{GJ56} Gel'fand, I. M. and Yaglom, A. M. (1956) Integration in function spaces and its application to quantum physics. (Russian) \emph{Uspehi Mat. Nauk (N.S.)} \textbf{11}, 77-–114.

\bibitem{ITO60} It\^o, K. (1960) Wiener integral and Feynman integral. \emph{Proc. Fourth Berkeley Symp. on Math. Stat. Probab.} II, 228--238.

\bibitem{KAC49} Kac, M. (1949) On distributions of certain Wiener functionals. \emph{Trans. A.M.S.} \textbf{65}, 1–-13.

\bibitem{KAC51} Kac, M. (1951) On some connections between probability theory and differential and integral equations. \emph{Proc. Second Berkeley Symp. on Math. Stat. Probab.}, 189–-215.

\bibitem{KEPE91} Kelley, W.G. and Peterson, A.C. (1991) \emph{Difference equations: An introduction with applications.} Academic Press, San Diego.

\bibitem{KNI62} Knight, F.B. (1962) On the random walk and Brownian motion. \emph{Trans. Amer. Math. Soc.}, \textbf{103}, 218--228.

\bibitem{KOL05} Kolokoltsov, V.N. (2005) Path integration: Connecting pure jump and Wiener processes. \emph{In: Probability and Partial Differential Equations in Modern Applied Mathematics, Eds.: Waymire, E.C. and Duan J.Q.}, \emph{The IMA Volumes in Math. and its Appl.}, \textbf{140}, 163--180.

\bibitem{MP10} M\"orters, P. and Peres, Y. (2010) \emph{Brownian Motion.} \emph{Cambridge Series in Statistical and Probabilistic Mathematics.} Cambridge University Press, Cambridge.

\bibitem{NEL64} Nelson, E. (1964) Feynman integrals and the Schr\"odinger equation. \emph{J. Math. Phys.} \textbf{5}, 332--343.

\bibitem{NSZ16} Nika, Z. and Szabados, T. (2016) Strong approximation of Black--Scholes theory based on simple random walks. \emph{Studia Sci. Math. Hung.}, \textbf{53}, 93--129.

\bibitem{PWZ97} Petkovsek, M., Wilf H.S. and Zeilberger, D. (1997) $A=B$ \url{http://www.math.upenn.edu/\~wilf/AeqB.html}

\bibitem{REV90} R\'ev\'esz, P. (1990) \emph{Random walks in random and non-random environments.} World Scientific, Singapore.

\bibitem{SZA96} Szabados, T. (1996) An elementary introduction to the Wiener process and stochastic integrals. \emph{Studia Sci. Math. Hung.}, \textbf{31}, 249--297.

\bibitem{SZSZ09} Szabados, T. and Sz\'ekely, B. (2009) Stochastic integration based on simple, symmetric random walks. \emph{J. Theor. Probab.}, \textbf{22}, 203--219.

\bibitem{SZA16} Szabados, T. (2016) A simple wide range approximation of symmetric binomial distributions. Preprint arXiv:1612.01112 .

\end{thebibliography}
\end{document}